\documentclass[11pt]{article}
\usepackage{natbib}
\usepackage{booktabs} 
\usepackage[ruled]{algorithm2e} 
\usepackage{times}
\usepackage{soul}
\usepackage{url}
\usepackage[utf8]{inputenc}
\usepackage{graphicx}
\usepackage{amsmath}
\usepackage{amsthm}
\usepackage[switch]{lineno}
\usepackage{comment}
\usepackage{xcolor}
\usepackage{multirow}
\usepackage{makecell}
\usepackage{tabu}
\usepackage{enumitem}
\usepackage{hyperref}
\usepackage{cleveref}
\usepackage{amsfonts}
\usepackage{lipsum}
\usepackage{balance} 
\usepackage{blindtext}
\usepackage{algpseudocode}
\usepackage[english]{babel}
\usepackage{array}
\usepackage{authblk}
\usepackage{float}

\SetAlFnt{\small}
\SetAlCapFnt{\small}
\SetAlCapNameFnt{\small}
\SetAlCapHSkip{0pt}
\IncMargin{-\parindent}

\usepackage[bibliography=common]{apxproof}

\newcommand{\BibTeX}{\rm B\kern-.05em{\sc i\kern-.025em b}\kern-.08em\TeX}

\newcommand{\wmin}{w_{\min}}
\newcommand{\wmax}{w_{\max}}
\newcommand{\ww}{\mathbf{w}}
\newcommand{\gcdw}{\operatorname{gcd}(\ww)}

\usepackage{soul}
\usepackage{color}

\newcommand{\erel}[1]{}

\newtheoremrep{theorem}{Theorem}[section]
\newtheoremrep{proposition}[theorem]{Proposition}
\newtheoremrep{lemma}[theorem]{Lemma}
\newtheorem{observation}[theorem]{observation}

\theoremstyle{definition}
\newtheorem{example}[theorem]{Example}
\newtheorem{definition}[theorem]{Definition}
\newtheorem{remark}[theorem]{Remark}

\setcitestyle{authoryear}

\makeatletter
\newlength{\phaserulewidth}

\makeatother

\makeatletter

\title{Whoever Said Money Won't Solve All Your Problems? \\ \Large Weighted Envy-free Allocation with Subsidy\footnote{This article combines two works (\cite{aziz2024weighted, elmalem2024weighted}), both set to appear as extended abstracts at AAMAS.}}

\author[1]{Noga Klein Elmalem}
\author[3]{Haris Aziz}
\author[1]{Rica Gonen}
\author[4]{Xin Huang}
\author[4]{Kei Kimura}
\author[4]{Indrajit Saha}
\author[2]{Erel Segal-Halevi}
\author[4]{Zhaohong Sun}
\author[3]{Mashbat Suzuki}
\author[4]{Makoto Yokoo}
\affil[1]{The Open University of Israel}
\affil[2]{Ariel University, Israel}
\affil[3]{UNSW Sydney, Australia}
\affil[4]{ Kyushu University, Japan}

\begin{document}

\maketitle
\begin{abstract}
We explore solutions for fairly allocating indivisible items among agents assigned weights representing their entitlements.
Our fairness goal is \textbf{weighted-envy-freeness (WEF)}, where each agent deems their allocated portion relative to their entitlement at least as favorable as any other’s relative to their own.
Often, achieving WEF necessitates monetary transfers, which can be modeled as third-party subsidies. The goal is to attain WEF with bounded subsidies.

Previous work relied on characterizations of unweighted envy-freeness (EF), that fail in the weighted setting. This makes our new setting challenging. 
We present polynomial-time algorithms that compute WEF allocations with a guaranteed upper bound on total subsidy for monotone valuations and various subclasses thereof.

We also present an efficient algorithm to compute a fair allocation of items and money, when the budget is not enough to make the allocation WEF. This algorithm is new even for the unweighted setting.
\end{abstract}

\section{Introduction}
\label{sec:intro}
The mathematical theory of fair item allocation among agents has practical applications, such as inheritance and partnership dissolutions. When agents have equal entitlements, each expects a share at least as good as others', known as an \emph{envy-free (EF)} allocation.
For indivisible items, an EF allocation might not exist. A common solution is using \emph{money} to compensate for envy. Recent studies assume a hypothetical third-party provides a non-negative \emph{subsidy} for each agent, and focus on minimizing the \emph{total subsidy} needed for envy-freeness. 


\cite{HaSh19a} showed that for any allocation, there exists a permutation of bundles that is \emph{envy-freeable (EF-able)}, meaning it can be made EF with subsidies. They proved the required subsidy is at most $(n-1)mV$, where $m$ is the number of items, $n$ the number of agents, and $V$ the maximum item value, and this bound is tight when the allocation is given. \cite{Brustle2020} presented an algorithm using iterative maximum matching, that computes an EF-able allocation with subsidy at most $(n-1)V$, also tight.
We extend the study to agents with different entitlements, or \emph{weights}, as in partnership dissolutions, where agents hold varying numbers of shares. For example, an agent with twice the entitlement of another expects a bundle worth at least twice as much.

Formally, an allocation is \emph{weighted envy-free (WEF)} (see e.g., \cite{robertson1998cake,zeng2000approximate,CISZ21}) if, for any two agents $i$ and $j$, $\frac{1}{w_i}$ times the utility $i$ assigns to their own bundle is at least $\frac{1}{w_j}$ times the utility $i$ assigns to $j$'s bundle, where $w_i$ and $w_j$ are their entitlements. 


We define \emph{weighted envy-freeability (WEF-ability)} analogously to EF-ability: An allocation is WEF-able if it can be made WEF with subsidies. Specifically, for any two agents $i$ and $j$, $\frac{1}{w_i}$ times the sum of the utility that $i$ assigns to his own bundle and the subsidy he receives is at least as high as $\frac{1}{w_j}$ times the sum of the utility that $i$ assigns to the bundle of $j$ and the subsidy $j$ receives.
Here, we assume quasi-linear utilities.

To illustrate the challenges in the generalized setting of unequal entitlements, we demonstrate that the results from \cite{HaSh19a,Brustle2020} for additive valuations fail when agents have different entitlements. 

\begin{example}[No permutation of bundles is WEF-able]
\label{example_intro}
There are two items $o_1,o_2$ and two agents $i_{1}, i_{2}$, with weights $w_{1} = 1, w_{2} = 10$ 
and valuation functions 
$$
\begin{bmatrix}
   & o_1 & o_2 \\
  i_1 & 5 & 7 \\
  i_2 & 10 & 8 \\
\end{bmatrix}
$$
Consider the bundles $X_1 = \{o_1\}$ and $X_2 = \{o_2\}$, where $i_1$ receives $X_1$ and $i_2$ receives $X_2$. Let $p_1$ and $p_2$ represent the subsidies for $i_1$ and $i_2$, respectively. The utility of $i_1$ for their own bundle is $5 + p_1$, and for $i_2$'s bundle, it is $7 + p_2$. To satisfy WEF, we need:
$\frac{5+p_1}{1} \geq \frac{7+p_2}{10}$, which implies $p_2 \leq 43 + 10 p_1$. 
Similarly, for agent $i_2$, WEF requires: $\frac{8+p_2}{10} \geq \frac{10+p_1}{1}$, which implies $p_2 \geq 92 + 10 p_1$. These two conditions are contradictory, so no subsidies can make this allocation WEF.
Next, consider the permutation where $i_1$ receives $X_2$ and $i_2$ receives $X_1$. In this case, WEF requires $\frac{7 + p_1}{1} \geq \frac{5+p_2}{10}$, which implies $p_2 \leq 65 + 10 p_1$, 
and for $i_2$: $\frac{10 +p_2}{10} \geq \frac{8 + p_1}{1}$, which implies $p_2 \geq 70 + 10 p_1$. Again, these conditions are contradictory, proving that no permutation of bundles is WEF-able.
\end{example}
\Cref{example_intro} also implies that the \emph{Iterated Maximum Matching algorithm} (\cite{Brustle2020}) does not guarantee WEF, as that algorithm yields an allocation where all agents receive the same number of items.

When valuations are not additive, even more results from the unweighted setting fail to hold.

\begin{example}
[Welfare-maximizing allocation is not WEF-able]
\label{ex:incompatibility-NW-WEFable}
There are two agents with weights $w_1=1, w_2=3$. There are two identical items. 
The agents have unit demand: agent 1 values any bundle with at least one item at $30$, and agent 2 at $90$. 
We show that, contrary to the result in \cite{HaSh19a}, the allocation maximizing the social welfare (sum of utilities) is not WEF-able.

The social welfare is maximized by allocating one item to each agent. Note that this is also the only allocation that is non-wasteful (\Cref{def:non-wasteful}).
WEF requires $\frac{30+p_1}{1}\geq \frac{30+p_2}{3}$, which implies $p_2\leq 3 p_1 + 60$; and $\frac{90+p_2}{3}\geq \frac{90+p_1}{1}$, 
which implies $p_2\geq 3 p_1 + 180$ --- a contradiction.
\end{example}
These examples demonstrate that the weighted case is more challenging than the unweighted case and requires new ideas.

Of course, since the unweighted case is equivalent to the weighted case where each weight $w_i = 1/n$, all
negative results from the unweighted setting extend to the weighted case. In particular, it is NP-hard to compute the minimum subsidy required to achieve (weighted) envy-freeness, even in the binary additive case, assuming the allocation is non-wasteful (as shown in \cite[Corollary 1]{HaSh19a}). 
Thus, following previous work,
we develop polynomial-time algorithms that, while not necessarily optimal, guarantee an upper bound on the total subsidy. 

\subsection{Our Results}
Our subsidy bounds for achieving WEF allocations are expressed as functions of $n$ (number of agents), $m$ (number of items), $\wmin$ (smallest agent weight), 
and $V$ (maximum item value); see \Cref{sec:model} for definitions.

For \emph{general monotone valuations}, we show that a total subsidy of $\left(\frac{W}{\wmin} - 1\right)mV$ is sufficient for WEF, and that this bound is tight in the worst case (\Cref{sec:characterization}).

For \emph{supermodular} and \emph{superadditive} valuations, we further show that WEF can be attained simultaneously with maximizing the social welfare and truthfulness 
(\Cref{sec:supermdular_val}).

For \emph{additive valuations}, assuming all weights are integers, we further show an upper bound that is independent of $m$: it is $\frac{W - \wmin}{\gcdw}V$,
where $\gcdw$ is the greatest common divisor of all weights --- largest number $d$ such that $w_i/d$ is an integer for all $i\in N$ (\Cref{sec:general_additive_val}).
Our algorithm extends the one in \cite{Brustle2020}; when all entitlements are equal it guarantees the same upper bound $(n-1)V$.

We also study WEF relaxations introduced for the setting without subsidy: WEF$(x,y)$ \cite{chakraborty2022weighted} and WWEF$1$ \cite{chakraborty2021picking}. 
We prove that WEF-ability is incompatible with WWEF1 or with WEF$(x,y)$ for $x+y<2$, but show an algorithm for two agents that finds a WEF-able and WEF$(1,1)$ allocation (\Cref{sec:app_EF}).

For \emph{identical additive valuations} (\Cref{sec:identical_additive_val}), we compute a WEF-able and WEF$(0,1)$ allocation with total subsidy at most $(n-1)V$, which is tight even in the unweighted case.

For \emph{binary additive valuations}, we adapt the General Yankee Swap algorithm (\cite{viswanathan2023general}) to compute a WEF-able and WEF$(0,1)$ allocation with total subsidy at most $\frac{W}{\wmin} - 1$, reducing to $n-1$ for equal weights (\Cref{sec:binary_additive_val}). For matroidal valuations, we show a linear lower bound in $m$ (\Cref{sec:matroidal}).

For identical items, 
we derive an almost tight bound of $V \sum_{2 \leq i \leq n} \left(w_i \sum_{1 \leq j \leq i} \frac{1}{w_j}\right)$,
where agents are sorted by descending order of their value for a single-item (\Cref{sec:identical_items}).
In particular, with nearly equal but different weights, the required subsidy may be in $\Omega(n^2 V)$, unlike the $O(n V)$ bound for equal weights. Additionally, we present a polynomial-time algorithm for computing a WEF-able allocation that requires the smallest possible amount of subsidy for each specific instance.
This is in contrast to the other sub-cases of additive valuations (binary additive and identical additive), in which computing the minimum subsidy per instance is known to be NP-hard.

Finally, we address scenarios with limited subsidy budgets (e.g., leftover cash in partnership dissolutions), exploring relaxed fairness under subsidy constraints (\Cref{sec:MWEF}).

Our contributions are summarized in Table~\ref{tbl:results}.
\begin{table}[h!]
\caption{Upper and lower bounds on worst-case 
total subsidy in weighted envy-freeable allocations. All subsidy upper bounds are attainable by polynomial-time algorithms.
Here, $w_2$ represents the second-smallest weight.
}
\label{tbl:results}
\begin{center}
\resizebox{\columnwidth}{!}{
\begin{tabular}{|c|c|c|}
\hline 
\makecell{\textbf{Valuation}} &\textbf{Lower bound}& \textbf{Upper bound}\\
\hline 
\makecell{\textbf{General,}\\\textbf{superadditive},\\\textbf{supermodular}} &  $\left(\frac{W}{\wmin} - 1\right)mV$ 
& $\left(\frac{W}{\wmin} - 1\right)mV$\\
       &  [~\Cref{prop:subsidy-ub-general}]     & [~\Cref{prop:subsidy-ub-general}] \\ \hline
\vspace{-2mm}& & \\
\makecell{\textbf{Additive}} & $\left(\frac{W}{\wmin} - 1\right)V$ & $\frac{W-\wmin}{\gcdw}V$ \\ 
       &  [~\Cref{worst case allocation can be chosen with weights}]     & [~\Cref{cor: sub general additive}] \\\hline
\vspace{-2mm}& & \\
\makecell{\textbf{Identical additive}} &  $(n-1)V$ &  $(n-1)V$ \\
       &  [~\Cref{thm:subsidy-lb-identical-additive}]    & [~\Cref{thm:poly-on-subsidy-identical-additive}] \\\hline
\vspace{-2mm}& & \\
\makecell{\textbf{Binary additive}} & $\frac{W}{w_2} - 1$ & $\frac{W}{\wmin} - 1$ \\ 
       &  [~\Cref{prop:lower-bound-binary}]    & [~\Cref{app:theorem_21}] \\\hline
\vspace{-2mm}& & \\
\makecell{\textbf{Matroidal}} & $\frac{m}{n}\left(\frac{W}{\wmin} - n\right)$ 
& $\left( \frac{W}{\wmin} - 1\right)m$ \\ 
       &  [~\Cref{thm:subsidy-lb-matroidal}]     & [~\Cref{prop:subsidy-ub-general}] \\\hline
\vspace{-2mm}& & \\
\makecell{\textbf{Additive},\\\textbf{identical items}}& 
{$\displaystyle\sum_{2 \leq i \leq n} \left( V w_i \sum_{1 \leq j < i} \frac{1}{w_j} \right)$} 
 & 
{$\displaystyle\sum_{2\leq i \leq n} \left( V w_i \sum_{1\leq j \leq i} \frac{1}{w_j} \right)$}
\\
&  [~\Cref{thm:subsidy-lb-additive-identical-items}]     & [~\Cref{thm:subsidy-ub-additive-identical-items}] \\ 

\hline
\end{tabular}
}
\end{center}
\end{table}

In \Cref{sec:exp}, we present preliminary experiments on the required subsidy in random instances and compare 
our algorithms with theoretical bounds.

\subsection{Related Work}
\label{sec:related_work}
 \paragraph{\textbf{Equal entitlements.}}
 \cite{steinhaus1948problem} initiated fair allocation with the cake-cutting problem, followed by \cite{foley1966resource} advocacy for envy-free resource allocation. Challenges with indivisible items were outlined by \cite{schmeidler1971fair}. 
 
 \paragraph{\textbf{Fair allocation with Monetary Transfers.}}
 The concept of compensating an indivisible resource allocation with money has been explored in the literature ever since 
 \cite{demange1986multi} introduced an ascending auction for envy-free allocation using monetary payments for unit demand agents.
 
 Recently, the topic of multi-demand fair division with subsidies has attracted significant attention (see, e.g., a survey article \cite{LLSW24}).
 \cite{HaSh19a} showed that
 an allocation is envy-freeable with money if and only if the agents cannot increase
 social welfare by permuting bundles.
 \cite{Brustle2020} study the more general class of monotone valuations. They demonstrate that a total subsidy of $2(n-1)^2   V$ is sufficient to guarantee the existence of an envy-freeable allocation. Moreover, they showed that 
 for additive valuations where the value of each item is at most $V$,
 giving at most $V$ to each agent (i.e., a total subsidy of at most $(n-1)V$) is sufficient to eliminate envies.
 \cite{kawase2024towards} improved this bound to $\frac{n^2 - n - 1}{2}$.
 
 \cite{caragiannis2020computing} developed an algorithm that approximates the minimum subsidies with any required accuracy for a constant number of agents, though with increased running time. However, for a super-constant number of agents, they showed that minimizing subsidies for envy-freeness is both hard to compute exactly and difficult to approximate.
 
 \cite{Aziz21a} presented a sufficient condition and an algorithm to achieve envy-freeness and equitability (every agent should get the same utility) when monetary transfers are allowed for agents with quasi-linear utilities and superadditive valuations (positive or negative).
 
 \cite{barman2022achieving} studied agents with  dichotomous valuations (agents whose marginal value for any good is either zero or one), without any additivity requirement.
 They proved that, for $n$ agents, there exists an allocation that achieves envy-freeness with total required subsidy of at most $n-1$, which is tight even for additive valuations.
 
 \cite{Goko2024} study an algorithm for an envy-free allocation with subsidy, that is also \emph{truthful}, when agents have submodular \emph{binary} valuations. The subsidy per agent is at most $V=1$.
 Their algorithm works only for agents with equal entitlements.
 
 The case where multiple items can be allocated to each agent while the agents pay some amount of money to the mechanism designer, is extensively studied in combinatorial auctions \citep{cramton2006combinatorial}. A representative mechanism is the well-known Vickrey-Clarke-Groves (VCG)  mechanism (~\cite{clarke,groves:econometrica:1973,vickrey:1961}), 
 which is truthful and maximizes social welfare. Envy-freeness is not a central issue in combinatorial auctions, with a notable exception presented by \cite{Papa03b}.
 
 \paragraph{\textbf{Different entitlements.}}
 In the past few years, several researchers have examined a more general model in which different agents may have different entitlements, included weighted fairness models like \emph{weighted maximin share fairness (WMMS)} and \emph{weighted proportionality up to one item (WPROP1)} (\cite{chakraborty2021picking,BEF23a,ACL19a}).
 \cite{CISZ21} established \emph{maximum weighted Nash welfare (MWNW)} satisfies \emph{Pareto optimality} and introduced a weighted extension of EF1. \cite{suksompong2022maximum} demonstrated MWNW properties under binary additive valuations and its polynomial-time computability. They further extended these findings to various valuation types. 
 
 \paragraph{\textbf{Different entitlements with subsidies.}}
 There is a long-standing tradition in fair division to revisit settings and extend them to the case of weighted entitlements. In a classic book by \cite{BrTa96a}, many algorithms and results are extended to the cases of weighted entitlements. This tradition continues in the context of the allocation of indivisible items. 
 \cite{WuZZ23} presented a polynomial-time algorithm for computing 
 a PROP allocation of \emph{chores} among agents with additive valuations, with total subsidy at most $\frac{n  V}{4}$, which is tight.
 For agents with different entitlements, they compute a WPROP allocation with total subsidy at most $\frac{(n-1)  V}{2}$. 
 In a subsequent work (\cite{wu2024tree}), they further improved this bound to $(\frac{n}{3} - \frac{1}{6})V$. 
 
 As far as we know, weighted envy-freeness with subsides has not been studied yet. Our paper aims to fill this gap.
 
 \paragraph{\textbf{Matroid-rank valuations.}}
 Recent studies have considered \emph{matroid-rank} valuation (binary submodular). \cite{montanari2024weighted} introduce a new family of weighted envy-freeness notions based on the concept of \emph{transferability} and provide an algorithm for computing transferable allocations that maximize welfare. \cite{babaioff2021fair} design truthful allocation mechanisms that maximize welfare and are fair. Particularly relevant to our work is a recent work by \cite{viswanathan2022yankee}, who devised a fair allocation method inspired by Yankee Swap, achieving efficient and fair allocations when agents have submodular binary valuations. We use some of their techniques in our algorithms. Later, \cite{viswanathan2023general} generalize the Yankee Swap algorithm to efficiently compute allocations that maximize any fairness objective, called General Yankee Swap. 

\section{Model}
\label{sec:model}
\paragraph{\textbf{Agents and valuations.}}
Let $[t] = \{1, 2, \ldots, t \}$ for any positive integer $t$. 
We consider $n$ agents in $N$ and $m$ items in $M$.
Each agent $i\in N$ has a valuation function $v_i:2^M \rightarrow \mathbb{R}^{+}_{0}$,
specifying a non-negative real value $v_i(A)$ for a given bundle $A\subseteq M$. 
We write $v_{i}(o_1 ,...,o_t)$ instead of $v_{i}(\{o_1, ..., o_t\})$.

We assume that the valuation functions $v_i$ are \textit{normalized}, that is $v_i(\emptyset) = 0$ for all $ i \in N$ (\cite{suksompong2023weighted}), and \emph{monotone}, i.e., for each $i\in N$ and $A \subseteq B \subseteq M$,  $v_i(A)\le v_i(B)$.
We denote $V := \displaystyle\max_{i\in N, A \subseteq M, A\neq \emptyset} \frac{v_i(A)}{|A|}$.
Note that when the valuations are additive, $V$ equals to $\displaystyle\max_{i\in N, o\in M} v_i(o)$.%
\footnote{In previous work it was assumed that $V=1$.}
 Also, $\displaystyle\max_{i\in N} v_i(M) \leq m\cdot V$.

An \emph{allocation} $X=(X_1,\ldots, X_n)$ is a partitioning of the items into $n$ bundles where $X_i$ is the bundle allocated to agent $i$.
We assume allocation $X$ must be \emph{complete}, i.e., 
$\bigcup_{i \in N} X_i = M$ holds; 

An \emph{outcome} is a pair consisting of the allocation and the subsidies received by the agents, that is, a pair $(X,\mathbf{p})$ 
is the allocation that specifies bundle $X_i\subseteq M$ for agent $i$ and $\mathbf{p}\in (\mathbb{R}^+_0)^n$ specifies the subsidy $p_i$ received by agent $i$.

An agent $i$'s \emph{utility} for a bundle-subsidy pair $(X_j,p_j)$ is $v_i(X_j)+p_j$. In other words, we assume quasi-linear utilities.

\paragraph{\textbf{Envy.}}
An outcome $(X,\mathbf{p})$  is \emph{envy-free (EF)} if for all $i,j\in N$, it holds that $$v_i(X_i)+p_i\geq v_i(X_j)+p_j.$$ 
An allocation $X$ is \emph{envy-freeable (EF-able)} if there exists a subsidy vector $\mathbf{p}$ such that $(X,\mathbf{p})$ is EF. 

\paragraph{\textbf{Entitlements.}}
Each agent $i \in N$ is endowed with a fixed \emph{entitlement}  $w_i \in \mathbb{R}_{>0}$. We also refer to entitlement as \emph{weight}.
We assume, without loss of generality, that the entitlements are ordered in increasing order, i.e., $\wmin = w_1 \leq w_2 \leq \ldots \leq w_n = \wmax$. 
We denote $W := \sum_{i\in N} w_i$. 

\begin{definition}[Weighted envy-freeability]
An outcome $(X,\mathbf{p})$ is \emph{weighted envy-free (WEF)} 
if for all $i, j \in N$: 
$$\frac{v_i(X_i)+p_i}{w_i}\geq \frac{v_i(X_j)+p_j}{w_j}.$$

An allocation $X$ is \emph{weighted envy-freeable (WEF-able)} if there is subsidy vector $\mathbf{p}$, such that $(X,\mathbf{p})$ is WEF. 
\end{definition}
The term $\frac{v_i(X_i)}{w_{i}}$ represents the value per unit entitlement for agent $i$ in their allocation. The WEF condition ensures that this value is at least as high as $\frac{v_i(X_j)}{w_{j}}$, denoting the corresponding value per unit entitlement for agent $j$ in the same allocation.
WEF seamlessly reduces to envy-free (EF) concept when entitlements are equal, i.e., $w_i = w_j$ for all $i,j \in N$.

\paragraph{\textbf{Efficiency concepts.}}
\begin{definition}[Pareto efficiency]
We say allocation $X$ dominates another allocation $X'$ if 
$\forall i\in N, v_i(X_i)\geq v_i(X'_i)$ and 
$\exists j\in N$, $v_j(X_j) > v_j(X'_j)$ hold. 
We say $X$ is Pareto efficient 
if it is not dominated by any other allocation. 
\end{definition}
\begin{definition}[Non-wastefulness]
\label{def:non-wasteful}
We say allocation $X$ is \emph{non-wasteful} if $\forall i \in N$, $\forall o \in X_i$, if 
$v_i(X_i)= v_i(X_i \setminus \{o\})$ holds, then $v_j(X_j \cup  \{o\})=v_j(X_j)$ for all $j\neq i$. 
\end{definition}
In other words, no item can be transferred from one agent to another, and the transfer results in a Pareto improvement. 

For 
$S\subseteq N$ and 
$A \subseteq M$, let ${\mathcal X}^{S,A}$ denote all possible allocations of $A$ among $S$.
For an allocation $X$ and a subset $S\subseteq N$ of agents,  the \emph{social welfare} among the agents in $S$ is defined as $SW^{S}(X) := \sum_{i\in S} v_i(X_i)$. 

\begin{definition}[Maximizing social welfare allocation (MSW)]
We say allocation $X^{S,A}\in {\mathcal X}^{S,A}$ maximizing social welfare with respect to $S$ and $A$ ($MSW^{S,A}$) if 
$SW^{S}(X^{S,A}) \geq  SW^{S}(X)$
 holds for any $X \in {\mathcal X}^{S,A}$ \footnote{If $S = N$ and $A=M$, we omit ``with respect to $N$ and $M$'' and just say a maximizing social welfare (MSW) allocation.}.
\end{definition}

Any MSW allocation is Pareto efficient, 
but not vice versa. 

Weighted envy-freeability and non-wastefulness are generally incompatible (\Cref{ex:incompatibility-NW-WEFable}). To address this, we propose a weaker efficiency property:
\begin{definition}[Non-zero social welfare]
We say allocation $X$ satisfies non-zero social welfare property if
$SW^N(X)=0$, then for any other allocation $X'$, $SW^N(X')=0$ holds. 
\end{definition}
This property allows choosing an allocation $X$ such that $SW^N(X)=0$ only if social welfare is zero for all allocations.

Pareto efficiency implies both non-wastefulness and non-zero social welfare, but the reverse is not true. Non-wastefulness and non-zero social welfare are independent properties.

\section{Characterization for General Monotone Valuations}
\label{sec:characterization}
In this section we allow agents to have arbitrary monotone valuations.
We give a characterization of WEF-able allocations. For the characterization, we generalize a couple of previously studied mathematical objects to the weighted case. 

\begin{definition}[Weighted reassignment-stability]
We say that an allocation $X$ is \emph{weighted reassignment-stable} if 
\begin{align}\label{eq:WRS}
\sum_{i\in N} \frac{v_i(X_i)}{w_i}\geq \sum_{i\in N} \frac{v_i(X_{\pi(i)})}{w_{\pi(i)}}
\end{align}
for all permutations $\pi$ of $N$. 
\end{definition}
With equal weights, reassignment-stability implies that $X$ maximizes the sum of utilities. But with different weights, it does \emph{not} imply that $X$ maximizes any function.

\begin{definition}[Weighted envy-graph]
For any given allocation $X$, the corresponding \emph{weighted envy-graph}, $G_{X,w}$,
is a complete directed graph with vertex set $N$, each assigned a weight  $w_{i}$.

For any pair of agents $i,j\in N$, $cost_X(i,j)$ is the cost of edge $(i,j)$ in $G_{X,w}$, 
which presents the fact that agent $i$ has envy toward agent $j$ under the allocation $X$:  
$$cost_X(i,j) \ =\  \frac{v_i(X_j)}{w_j}-\frac{v_i(X_i)}{w_i}.$$
Note that $cost_X$ can take negative values.
For any path or cycle $C$ in the graph, $cost_X(C)$ is the cost of the $C$ under allocation $X$, which is the sum of costs of edges along $C$ in $G_{X,w}$.

With these definitions, $\ell_{i,j}(X)$ 
 represents the cost of the maximum-cost path from $i$ to $j$ in $G_{X,w}$, and $\ell_i(X) = cost_{X}(P_i(X))$ represents cost of the maximum-cost path in $G_{X,w}$ starting at $i$, denoted as $P_i(X)$.
\end{definition}

Similar to the previous work of \cite{HaSh19a} (in the unweighted setup), we provide necessary and sufficient conditions for a WEF-able allocation:
\begin{theorem}
\label{th:ef} The following are equivalent for  allocation $X$:
\begin{enumerate}
\item $X$ is WEF-able;  \label{th:ef:condition1}
\item $X$ is weighted reassignment-stable; \label{th:ef:condition2}
\item The graph $G_{X,w}$ has no positive-cost cycle. 
\label{th:ef:condition3}
	\end{enumerate}
	\end{theorem}
\begin{proof}
$1 \Rightarrow 2$. Suppose the allocation $X$ is WEF-able. Then, there exists a subsidy vector $\mathbf{p}$ such that for all agents $i,j$ 
$\frac{v_i(X_i)+p_i}{w_i}\geq \frac{v_i(X_j)+p_j}{w_j}$. Equivalently, 
$\frac{v_i(X_j)}{w_j}- \frac{v_i(X_i)}{w_i}\leq \frac{p_i}{w_i} - \frac{p_j}{w_j}$.
Consider any permutation $\pi$ of $N$. Then,
$$\sum_{i\in N} \left(\frac{v_i(X_{\pi(i)})}{w_{\pi(i)}}- \frac{v_i(X_i)}{w_i}\right)\leq \sum_{i\in N}\bigg(\frac{p_i}{w_i} - \frac{p_{\pi(i)}}{w_{\pi(i)}}\bigg)=0.$$ The last entry is zero as all the weighted subsidies are considered twice, and they cancel out each other. 
Hence the allocation $X$ is weighted reassignment-stable.


$2 \Rightarrow 3$.
Suppose some allocation $X$ has a corresponding weighted envy-graph with a cycle $C = (i_1, \ldots, i_r)$ of strictly positive cost. Then consider a permutation $\pi$, defined for each agent $i_k \in N$ as follows: \[\pi(i_k) = \begin{cases}
    i_k, &i_k\notin C \\
    i_{k+1}, & k \in \{1, \ldots, r-1 \} \\
    i_1 & k = r 
\end{cases}. \]
 In that case 
 $0 < cost_X(C) \Leftrightarrow \sum_{i\in N} \frac{v_i(X_i)}{w_i}< \sum_{i\in N} \frac{v_i(X_{\pi(i)})}{w_{\pi(i)}}$, which means that $X$ is not weighted reassignment-stable. 


$3 \Rightarrow 1$.
Suppose (3) holds. 
As there are no positive-weight cycles in $G_{X,w}$, we can define, for each agent $i$,
the maximum cost of any path in the weighted envy-graph that starts at $i$.
We denote this path by $\ell_i(X)$. 
Let each agent $i$'s subsidy be $p_i=\ell_i(X)\cdot w_i$.
Then for any other agent $j\neq i \in N$, 
$\frac{p_i}{w_i}=\ell_i(X)\geq cost_X(i,j)+\ell_j(X)=\frac{v_i(X_j)}{w_j}-\frac{v_i(X_i)}{w_i} + \frac{p_j}{w_j}.$
This implies that $(X,\mathbf{p})$ is WEF, and hence $X$ is WEF-able. 
\end{proof}




\Cref{th:ef} presents an effective method for verifying whether a given allocation is WEF-able. 

\begin{proposition}
Given an allocation $X$, it can be checked in polynomial time whether $X$ is WEF-able.
\end{proposition}
\begin{proof}
    According to \Cref{th:ef}, determining whether $X$ is is WEF-able is equivalent to verifying if $X$ is weighted reassignment-stable. To analyze this, consider a complete bipartite graph $G$ with $N$ nodes on each side. The weight of each edge connecting $i$ on the left to $j$ on the right is defined as $\frac{v_i(X_j)}{w_j}$.  This graph can be constructed in $O(mn + n^2)$ time. 
    Next, the maximum-value matching for $G$ can be computed in $O(n^3)$ time  (\cite{edmonds1972theoretical}). 
    If the value of the matching that contains edge $(i,i)$ for all $i\in N$ is equal to the maximum value, it follows that $\sum_{i\in N} \frac{v_i(X_i)}{w_i} \geq \sum_{i\in N} \frac{v_i(X_\pi(i))}{w_{\pi(i)}}$ for any permutation $\pi:N\rightarrow N$. This condition aligns with the definition of weighted reassignment-stability. The total running time is $O(mn + n^3)$.

    Another approach involves verifying the absence of positive-cost cycles in the weighted envy-graph $G_{X,w}$. This can be achieved by transforming the graph by negating all edge weights and then checking for the presence of negative-cost cycles. Using the Floyd-Warshall algorithm 
    (\cite{weisstein2008floyd,wimmer2017floyd}) on the graph obtained by negating all edge cost in $G_{X,w}$ requires $O(n^3)$ time. Constructing the initial graph $G_{X,w}$ takes $O(mn)$ time, resulting in an overall complexity of $O(mn + n^3)$.   
\end{proof}

Deciding that an allocation $X$ is WEF-able is not enough. we also need to find a minimal subsidy vector, $\mathbf{p}$, that ensures WEF for $(X, \mathbf{p})$.
The following theorem is similar to \cite[Theorem 2]{HaSh19a}, which states the minimum subsidy required when given a WEF-able allocation.

\begin{theorem}\label{thm:minsubsidy}
    For any WEF-able allocation $X$, let $\mathbf{p}^{*}$ be a subsidy vector defined by 
$p^{*}_{i} := w_i \ell_i(X)$,
for all $i\in N$. 
Then
    \begin{enumerate}
    \item $(X,\mathbf{p}^{*})$ is WEF;
    \item Any other subsidy vector $\mathbf{p}$, such that $
    (X,\mathbf{p})$ is WEF, 
    satisfies $p^{*}_{i} \leq p_{i}$ for all $i\in N$;
    \item $\mathbf{p}^{*}$ can be computed in $O(nm+n^{3})$ time;
    \item There exists at least one agent $i\in N$ for whom $p^*_i=0$. \label{agent with 0 subsidy}
\end{enumerate}
\end{theorem}
\begin{proof}
\begin{enumerate}
    \item The establishment of condition \ref{th:ef:condition3} implying condition \ref{th:ef:condition1} in  \Cref{th:ef} has already demonstrated that $(X,\mathbf{p}^{*})$ is WEF. 
    \item Let $\mathbf{p}$ be a subsidy vector, such that $(X,\mathbf{p})$ is WEF, and $i \in N$ be fixed. Consider the highest-cost path originating from $i$ in the graph $G_{X,w}$,  $P_i(X)=(i = i_{1}, ..., i_{r})$, with $cost_{X}(P_i(X)) = \ell_i(X) = \frac{p^{*}_{i}}{w_i}$.
   
    Due to the WEF nature of $(X,\mathbf{p})$, it follows that for each $k \in [r-1]$, the following inequality holds:
    \begin{align*}
        & \frac{v_{i_{k}}(X_{i_{k}})+ p_{i_{k}}}{w_{i_{k}}}  \geq \frac{v_{i_{k}}(X_{i_{k+1}}) + p_{i_{k+1}}}{w_{i_{k+1}}}
     \Rightarrow \frac{p_{i_{k}}}{w_{i_{k}}} - \frac{p_{i_{k+1}}}{w_{i_{k+1}}} \geq \\ &\frac{v_{i_{k}}(X_{i_{k+1}})}{w_{i_{k+1}}} - \frac{v_{i_{k}}(X_{i_{k}})}{w_{i_{k}}} = cost_X(i_{k}, i_{k+1}).
    \end{align*}
    Summing this inequality over all $k \in [r-1]$, the following relation is obtained:
    \begin{align*}
   &\frac{p_{i}}{w_{i}} - \frac{p_{i_{r}}}{w_{i_{r}}} = \frac{p_{i_{1}}}{w_{i_{1}}} - \frac{p_{i_{r}}}{w_{i_{r}}}\geq cost_X(P_i(X)) = \frac{p_{i}^{*}}{w_i}
   \Rightarrow \frac{p_{i}}{w_{i}} \geq \frac{p_{i}^{*}}{w_i} + \frac{p_{i_{r}}}{w_{i_{r}}} \geq \frac{p_{i}^{*}}{w_i}.
   \end{align*}
    The final transition is valid due to the non-negativity of subsidies and weights, that is, $\frac{p_{i_{r}}}{w_{i_{r}}} \geq 0$. 
    \item The computation of $\mathbf{p}^{*}$ can be executed as follows: Initially, the Floyd-Marshall algorithm 
(\cite{weisstein2008floyd}, \cite{wimmer2017floyd}) is applied to the graph derived by negating all edge costs in $G_{X,w}$ (This has a linear time solution since there are no cycles with positive costs in the graph). Hence, determining the longest path cost between any two agents, accomplished in $O(nm+n^{3})$ time. Subsequently, the longest path starting at each agent is identified in $O(n^{2})$ time.
    \item Assume, for the sake of contradiction, that $p^*_i > 0$ for every agent $i \in N$. This implies that $\ell_i(X) > 0$, since $w_i > 0$. 
Starting from an arbitrary agent $i_1$,
we trace the highest-cost path starting at $i_1$ and arrive at some agent $i_2$; then we trace the highest-cost path starting at $i_2$ and arrive at some agent $i_3$; and so on. 
Because the number of agents is finite, eventually we will arrive at an agent we already visited. We will then have a cycle with positive cost in $G_{X,w}$, contradicting \Cref{th:ef}.

Therefore, there must be at least one agent whose subsidy under $\mathbf{p}^{*}$ is $0$.
\end{enumerate}
\end{proof}

Now we can find the minimum subsidy needed in the worst-case scenario for agents with different entitlements, whether the allocation is given or can be chosen. 
\begin{theorem} \label{worst case allocation is given with weights} 
For every weight vector
and every given WEF-able allocation $X$, 
letting $p_i := p^*_i = w_i \ell_i(X)$,
the total subsidy $\displaystyle \sum_{i\in N}p_i$
is at most $\left(\frac{W}{\wmin} - 1\right)   m   V$. This bound is tight in the worst case.
\end{theorem}

\begin{proof}
The proofs extend those of \cite{HaSh19a}. 
By \Cref{thm:minsubsidy}, to bound the subsidy required for $i$, we bound the highest cost of a path starting at $i$.
We prove that, for every WEF-able allocation $X$ and agent $i$, the highest cost of a path from $i$ in $G_{X,w}$ is at most $\frac{m   V}{\wmin}$.

For every path $P$ in $G_{X,w}$,
\begin{align*}
&cost_X(P) =
\sum_{(i,j)\in P} cost_X(i,j)  = \sum_{(i,j)\in P}
\frac{v_{i}(X_{j})}{w_{j}} - \frac{v_{i}(X_{i})}{w_{i}} 
 \leq \\ &\sum_{(i,j)\in P} \frac{v_{i}(X_{j})}{\wmin} \leq 
\sum_{(i,j)\in P}\frac{V\cdot |X_{j}|}{\wmin}
\leq 
 \frac{m V}{\wmin}.  
\end{align*} 
Therefore, the cost of every path is at most $\frac{m V}{\wmin}$, so agent $i$ needs a subsidy of at most $w_i \frac{m V}{\wmin}$. 
By part (4) of \Cref{thm:minsubsidy}, at least one agent has a subsidy of 0. This implies a total subsidy of at most $\frac{W - \wmin}{\wmin}  m   V 
   = \left(\frac{W}{\wmin} - 1\right) m V $.

To establish tightness,
   let us assume all agent has an all-or-nothing valuation for $M$, such that 
   such that $v_1(M) = mV$ and $v_i(M) = mV-\varepsilon$ for $i\neq 1$.
   Consider the allocation $X$, which assigns all items to a single agent $1$ (with the minimum entitlement $\wmin$). It's evident that: (1) $X$ is WEF-able, (2) satisfies non-zero social welfare property, (3) $X$ is MSW (and hence, Pareto efficient) and (4) its optimal subsidy vector $\mathbf{p}$ satisfies $p_1 = 0$ and $p_i = \frac{w_i}{\wmin} (mV-\varepsilon)$ for $i\neq 1$. Therefore, we require 
   $\frac{W - \wmin}{\wmin} (mV-\varepsilon) 
   = \left(\frac{W}{\wmin} - 1\right) (mV-\varepsilon)$ in total.
As $\varepsilon$ can be arbitrarily small, we get a lower bound of $\left(\frac{W}{\wmin} - 1\right) mV$.  
\end{proof}

In the unweighted case $W/\wmin=n$, so the upper bound on the subsidy becomes $(n-1)mV$.
This is the same upper bound proved by 
\cite{HaSh19a} for the unweighted case and additive valuations.

The following lemma is useful for showing weighed envy-freeability and the subsidy bounds. 
\begin{lemma}
\label{lem:max-value}
For an allocation $X$, if for all $i, j \in N$, 
$v_i(X_i)\geq v_j(X_i)$ holds, then $X$ is WEF-able, and the cost of any path from agent $i$ to $j$ is at most $\frac{v_j(X_j)}{w_j}-\frac{v_i(X_i)}{w_i}$.
\end{lemma}
\begin{proof}
\begin{enumerate}
\item 
It is sufficient to show that $X$ satisfies 
reassignment-stability. Indeed, for any permutation $\pi$, 
let $\pi^{-1}$ be the inverse permutation. Then the condition in the lemma implies
\begin{align*}
    \sum_{i\in N} \frac{v_i(X_i)}{w_i}
    &\ge
    \sum_{i\in N} \frac{v_{\pi^{-1}(i)}(X_i)}{w_i}.
\end{align*}
By re-ordering the summands in the right-hand side we get the reassignment-stability condition:
\begin{align*}
    \sum_{i\in N} \frac{v_i(X_i)}{w_i}
    &\ge
    \sum_{i\in N} \frac{v_i(X_{\pi(i)})}{w_{\pi(i)}}.
\end{align*}

\item For any path $P$ from $i$ to $j$ in the weighted envy-graph,
\begin{align*}
cost_X(P)&=\sum_{(h,k)\in P}cost_X(h,k)
=\sum_{(h,k)\in P} \frac{v_h(X_k)}{w_k}-\frac{v_h(X_h)}{w_h}
\le \\ & \sum_{(h,k)\in P}\frac{v_k(X_k)}{w_k}-\frac{v_h(X_h)}{w_h}.
\end{align*}
The latter expression is a telescopic sum that simplifies to the difference of two  elements:
$$\frac{v_j(X_j)}{w_j}-\frac{v_i(X_i)}{w_i}.$$
\end{enumerate}
\end{proof}
We now prove the main positive result of this section.
\begin{theorem}\label{prop:subsidy-ub-general}
For every instance with monotone valuations,
there exists a WEF-able, non-zero-social-welfare allocation 
with total subsidy at most $\left(\frac{W}{\wmin} - 1\right) mV$. 
This bound is tight in the worst case.
\end{theorem}
\begin{proof}
Let $i^*$ be an an agent for whom $v_{i^*}(M)$ is largest. Let $X$ be the allocation in which all items are allocated to $i^*$.
It is clear that this allocation satisfies non-zero social welfare property.
Also, $v_i(X_i) \geq v_j(X_i)$ holds for any $i, j \in N$. 
Thus, by \Cref{lem:max-value}, the allocation is WEF-able.
To eliminate envy, agent $i^*$ should receive no subsidy, while any other agent $i\neq i^*\in N$ should receive a subsidy of $\frac{w_i}{w_{i^*}}v_j(M) \leq \frac{w_i}{\wmin}v_j(M)$. 
For any agent $j$, $v_j(M) \leq m V$  by definition of $V$.
Therefore, the total subsidy is at most $\left(\frac{W}{\wmin} - 1\right) mV$.

To prove tightness, we use the same example as in the tightness proof of \Cref{worst case allocation is given with weights}.
Every non-zero-social-welfare allocation must allocate all items to a single agent.
Therefore, the situation is identical to allocating a single item. 
The \textbf{only} WEF-able allocation is the one giving all items to agent assigning the highest value to $M$, who is agent $1$ (with the lowest entitlement). Therefore, the subsidy must be $\left(\frac{W}{\wmin} - 1\right) (mV-\epsilon)$ for any $\epsilon>0$.
\end{proof}

Importantly, the maximum worst-case subsidy in the weighted setting depends on the proportion of $\wmin$ in the total weight $W$, which can be much larger than the number of agents $n$ used in the unweighted setting bounds.


\Cref{prop:subsidy-ub-general} guarantees only a very weak efficiency notion: non-zero-welfare.
We do not know if WEF is compatible with Pareto-efficiency for general monotone valuations. However, for the large sub-class of 
\emph{superadditive valuations}, we prove in the following section that every allocation maximizing the sum of utilities (\emph{maximizing social welfare} or MSW) is WEF-able. In particular, a Pareto-efficient WEF allocation always exists. 
Moreover, we prove that such an allocation can always be attained by a \emph{truthful mechanism} ---  which induces agents to reveal their true utility functions. 

\section{Superadditive and Supermodular Valuations}
\label{sec:supermdular_val}
A monotone valuation $v_i$ is called 
\begin{itemize}
    \item \emph{Superadditive} ---
if for any $X, Y\subseteq M$ with $X \cap Y = \emptyset$, 
$v_i(X) + v_i(Y) \leq v_i(X \cup Y)$.
    \item \emph{Supermodular} ---
if for any $X, Y\subseteq M$ holds
$v_i(X) + v_i(Y) \leq v_i(X \cup Y) + v_i(X \cap Y)$.



\end{itemize}
Every additive valuation is supermodular, and every supermodular valuation is superadditive.

The lower bound for monotone valuations (in the proof of \Cref{prop:subsidy-ub-general}) is attained by supermodular valuations, so it applies to supermodular and superadditive valuations too.

In this section we prove that, for superadditive valuations (hence also for supermodular and additive valuations), the same upper bound of \Cref{prop:subsidy-ub-general} can be attained by an allocation that maximizes the social welfare. 

\begin{theorem}
When valuations are superadditive, any MSW allocation is WEF-able.
\end{theorem}
\begin{proof}
We show that for an MSW allocation $X$, $v_i(X_i) \geq v_j(X_i)$ holds for any $i, j \in N$. For the sake of contradiction, assume $v_i(X_i) < v_j(X_i)$ holds. In this case, we can construct another allocation $X'$, 
where for all $k \neq i, j$, $X'_k=X_k$, 
$X'_i =\emptyset$, $X'_j = X_j \cup X_i$. In other words, we reassign $X_i$ from $i$ to $j$. By the definition of superadditivity, we have $v_j(X_j') = v_j(X_i \cup X_j) \geq v_j(X_i) + v_j(X_j)$. Therefore, the total social welfare under allocation $X'$ is
\begin{align*}
    & SW^N(X') = v_j(X_j') + \sum_{k\neq i,j} v_k(X_k) >   v_i(X_i) + v_j(X_j) + \sum_{k\neq i,j} v_k(X_k)= SW^N(X).
\end{align*}
 
However, this contradicts the fact that $X$ is a MSW allocation. 
From \Cref{lem:max-value}, it follows that $X$ is WEF-able. 
\Cref{worst case allocation is given with weights} implies that a subsidy of $\left(\frac{W}{\wmin} - 1\right)   m   V$ is sufficient.
\end{proof}
\Cref{ex:incompatibility-NW-WEFable} shows that the theorem does not hold without the superadditivity assumption.

A \emph{mechanism} is a function from the profile of 
declared agents' valuation functions to an outcome. 
We say a mechanism is truthful if no agent can obtain a 
strictly better outcome by misreporting its valuation function.

\begin{definition}[VCG mechanism (\cite{vickrey:1961, clarke,groves:econometrica:1973})]
The VCG mechanism chooses an allocation $X$ which 
maximizes $SW^N(X)$ among all allocations of $M$ to $N$.

Agent $i$, who is allocated $X_i$, pays a price equal to $SW^{N \setminus \{i\}}(X') 
- SW^{N \setminus \{i\}}(X)$,
where $X'$ maximizes $SW^{ N\setminus\{i\}}$ among all allocations of $M$ to $N\setminus\{i\}$.
\end{definition}
\begin{theorem}
\label{thm:supermodular_VCG}
When valuations are superadditive, 
the VCG mechanism with a large up-front subsidy 
(i.e.,we first distribute $C\cdot w_i$ to agent $i$, and if agent $i$ obtains 
a bundle, it pays the VCG payment from $C\cdot w_i$) is WEF, 
Pareto efficient, and truthful. 
\end{theorem}
\begin{proof} 
Truthfulness and Pareto efficiency are clear. 
We show that it is WEF. 
We first show that in the VCG, 
for each agent $i$ who obtains $X_i$ and pays $q_i$, 
$q_i \geq v_j(X_i)$ holds for any $j\neq i$. 
For the sake of contradiction, 
assume $q_i=  
SW^{N \setminus \{i\}}(X') 
- SW^{N \setminus \{i\}}(X)
< v_j(X_i)$ holds. 
Then, 
$SW^{N \setminus \{i\}}(X')  
< v_j(X_i)  + SW^{N \setminus \{i\}}(X)$ holds. 
However, if we consider an allocation of $M$ to agents except for $i$, 
we can first allocate $M\setminus X_i$ optimally among $N\setminus \{i\}$, 
then allocate $X_i$ additionally to agent $j$. Then, 
the total valuation of this allocation is at least 
$v_j(X_i)  + SW^{N \setminus \{i\}}(X)$ due to 
superadditivity. This contradicts the fact that 
$SW^{N \setminus \{i\}}(X')$ is the total valuation when allocating 
$M$ optimally among agents except for $i$. 
Also, it is known that VCG is individually rational, which means that $v_i(X_i) \geq q_i$ holds for all $i \in N$.

Combining both inequalities leads to

\begin{align*}
&
    \frac{v_j(X_j) + C\cdot w_j - q_j}{w_j}
       \geq
    \frac{C\cdot w_j}{w_j}
    = 
    C
    =
    \frac{C\cdot w_i}{w_i}
       \geq \frac{v_j(X_i) + C\cdot w_i - q_i}{w_i},
\end{align*}
which is the WEF condition.
\end{proof}
A similar mechanism is presented by \cite{Goko2024} for the unweighted case. 

To guarantee that all subsidies are non-negative, $C$ should be an upper bound on the payment of each agent. As the payment of each agent is at most the social welfare in an allocation, which is at most $mV$, we can simply take $C := m V / \wmin$.

\section{Additive Valuation}
\label{sec:general_additive_val}
The valuation function of an agent $i$ is called
\emph{Additive} if for each $i\in N$ and $A,B\subseteq M$ such that $A\cap B=\emptyset$, $v_i(A\cup B)= v_i(A)+v_i(B)$.
Without loss of generality, we assume that each item is valued positively by at least one agent; items that are valued at $0$ by all agents can be allocated arbitrarily without affecting the envy.
\cite{HaSh19a} prove that, with additive valuations, the minimum subsidy required in the worst case is at least $(n-1)   V$, even for binary valuations, when the allocation can be chosen.
We generalize their results as follows.
\begin{lemma} \label{lemm: allocation item to agent with highest value}
Suppose there are $n$ agents,
and only one item $o$ which the agents value positively. Then an allocation is WEF-able iff $o$ is given to an agent $i$ with the highest $v_i(o)$.
\end{lemma}

\begin{proof}
    By \Cref{th:ef}, it is sufficient to check the cycles in the weighted envy-graph.
If $o$ is given to $i$, then $i$'s envy is $- \frac{v_i(o)}{w_i}$,
and the envy of every other agent $j$ in $i$ is $\frac{v_j(o)}{w_i}$.
All other envies are 0.
The only potential positive-weight cycles are cycles of length 2 involving agent $i$.
The weight of such a cycle is positive iff $\frac{v_j(o)}{w_i} - \frac{v_i(o)}{w_i} >0$, which holds iff $v_j(o) > v_i(o)$.
Therefore, there are no positive-weight cycles iff $v_i(o)$ is maximum.
\end{proof}

\begin{theorem}
\label{worst case allocation can be chosen with weights}
For every weights $\mathbf{w}$ and $n\geq 2$,
no algorithm can guarantee a total subsidy smaller than $\left(\frac{W}{\wmin}-1\right)  V$.
\end{theorem}


\begin{proof}



Consider an instance with $n$ agents and one item $o$ with valuations $v_1(o) = V$ and $v_i(o) = V-\epsilon$ for $i\ge 2$.

By \Cref{lemm: allocation item to agent with highest value},
the only WEF-able allocation is to give $o$ to agent $1$.
In this case, the minimum subsidy is $p_1 = 0$ and $p_i = w_i   \frac{v_i(o)}{\wmin} = w_i \frac{\left(V-\epsilon\right)}{\wmin}$.
Summing all subsidies leads to $$\sum_{i\geq 2} w_i  \frac{\left(V-\epsilon\right)}{\wmin} = \frac{W-\wmin}{\wmin} \left(V-\epsilon\right) = \left(\frac{W}{\wmin} - 1\right) \left(V-\epsilon\right).$$
As $\epsilon$ can be arbitrarily small, we get a lower bound of $\left(\frac{W}{\wmin} - 1\right) V$.  
\end{proof}

In \Cref{example_intro}, we showed that the iterated-maximum-matching algorithm (\cite{Brustle2020}) might produce an allocation that is not WEF-able.

We now introduce a new algorithm, \Cref{alg:general-additive}, which extends the iterated-maximum-matching approach to the weighted setting, assuming all weights are integers. The algorithm finds a one-to-many maximum matching between agents and items, ensuring that each agent $i\in N$ receives exactly $w_i$ items.
If the number of items remaining in a round is less than $W$, we add dummy items (valued at $0$ by all agents) so that the total number of items becomes $W$.
In \Cref{example_intro}, we add 9 dummy items, 
and perform a one-to-many maximum-value matching between agent and items, resulting in a WEF-able allocation: $X_1 = \emptyset, X_2 = \{o_1, o_2\}$.

The algorithm runs in $\lceil m/W \rceil$ rounds.
In each round $t$, the algorithm computes a one-to-many maximum-value matching $\{ X_i^t \}_{i\in N}$ between all agents and unallocated items $O_t$, where each agent $i\in N$ receives exactly $w_i$ items. 

To achieve this, we reduce the problem to the \emph{minimum-cost network flow problem} (\cite{goldberg1989network}) by constructing a flow network and computing the maximum integral flow of minimum cost. The flow network is defined as follows:
\begin{itemize}
    \item \textbf{Layer 1 (Source Node).} a single source node $s$.
    \item \textbf{Layer 2 (Agents).} a node for each agent $i\in N$, with an arc from $s$ to $i$, having cost $0$ and capacity $w_i$.
    \item \textbf{Layer 3 (Unallocated Items).} a node for each unallocated item $o \in O_t$, with an arc from each agent $i\in N$ to item $o$, having cost $-v_i(o)$ and capacity $1$.
    \item \textbf{Layer 4 (Sink Node).} a single sink node $t$, with an arc from each item $o\in O_t$ to $t$, having 0 cost and capacity $1$.
\end{itemize}
Any integral maximum flow in this network corresponds to a valid matching where each agent $i\in N$ receives exactly $w_i$ items from $O_t$, and each item is assigned to exactly one agent, the result is a minimum-cost one-to-many matching based on the costs in the constructed network. Because we negate the original costs in our construction, the obtained matching $\{X_i^t\}_{i\in N}$ maximizes the total value with respect to the original valuations.
After at most $\lceil m/W \rceil$ valuations, all items are allocated.

\begin{algorithm}[h]
\SetAlgoNoLine
\caption{
\label{alg:general-additive}
Weighted Sequence Protocol For Additive Valuations and Integer weights}
\KwIn{Instance $(N,M,v, \mathbf{w})$ with additive valuations.}
\KwOut{WEF-able allocation $X$ with total required subsidy of at most $(W-\wmin)  V$.}
$X_{i} \gets \emptyset, \forall i\in N$\;
$t \gets 1$; $O_1 \gets M$\;
\While{$O_t \neq \emptyset$}{
        Construct the flow network $G'=(V',E')$:
        \begin{itemize}
            \item define $V'= N \cup O_t \cup \{s,t\}$.
            \item Add arcs with the following properties:
            \begin{itemize}
                \item From $s$ to each agent $i\in N$ with cost $0$ and capacity $w_i$.
                \item From each agent $i\in N$ to each unallocated item $o \in O_t$, with cost $-v_i(o)$ and capacity 1.
                \item From each unallocated item $o\in O_t$ to $t$ with cost $0$ and capacity $1$.
            \end{itemize}
        \end{itemize}
        Compute an \emph{integral maximum flow of minimum cost} on $G'$, resulting in the one to many matching $\{X_i^t\}_{i\in N}$\;
        \text{Set $X_i \gets X_i \cup X_i^t$, for all $i\in N$}\;
        \text{Set $O_{t+1} \gets O_t \backslash \cup_{i\in N}X_i^t$}\;
        $t \gets t + 1$
}
      \Return $X$
\end{algorithm}

\begin{proposition} \label{general additive: WEF-able}
     For each round $t$ in \Cref{alg:general-additive}, $X^t$ is WEF-able.
\end{proposition}



\begin{proof}
    We prove that, in every round $t$, the total cost added to any directed cycle in the weighted-envy graph is non-positive. Combined with \Cref{th:ef}, this shows that $X^t$ is WEF-able for every round $t \in T$.
    
    Let $X^t$ the allocation computed by \Cref{alg:general-additive} at iteration $t$. Note that \Cref{alg:general-additive} is deterministic. 
    Let $C$ be any directed cycle in $G_{X^t,w}$, and denote $C = (i_1,..., i_r)$. To simplify notation, we consider $i_1$ as $i_{r+1}$. 

    Given the allocation $X^t$ and the cycle $C$, we construct a random alternative allocation $B^t$ as follows: 
    for each agent $i_j \in C$, we choose one item $o_{i_{j+1}}^t$ uniformly from $i_{j+1}$'s bundle and transfer it to $i_{j}$'s bundle \footnote{Recall that at each iteration, each agent $i_j$ receives exactly $w_{i_j}$ items.}.
    The expected value of $v_{i_j}(o_{i_{j+1}}^t)$, the value of the item removed from $i_j$'s bundle, can be computed as the average value of all items in $X^t_{i_j}$: $\frac{\sum_{o\in X^t_{i_j}} v_{i_j}(o)}{w_{i_j}} = \frac{v_{i_j}(X^t_{i_j})}{w_{i_j}}$. Similarly, the expected value of $v_{i_j}(o_{i_j}^t)$, the value of the item added the $i_j$'s bundle, is $\frac{v_{i_j}(X^t_{i_{j+1}})}{w_{i_{j+1}}}$.
    Thus, the expected change in value between $B^t$ and $X^t$ is
    \begin{align*}
        & \mathbb{E}\left[\sum_{i\in N} \left( v_i(B_i^t) - v_i(X^t_i) \right)\right] = 
        \sum_{i_j \in C} \left( \mathbb{E}\left[v_{i_j}(o_{i_{j+1}}^t)\right] - \mathbb{E}\left[v_{i_j}(o_{i_j}^t)\right] \right) = \\ &
        \sum_{i_j\in C} \frac{v_{i_j}(X^t_{i_{j+1}})}{w_{i_{j+1}}} - \frac{v_{i_j}(X^t_{i_j})}{w_{i_j}}
    .
    \end{align*} 
    This is exactly the total cost of cycle $C$.
    According to \Cref{alg:general-additive}, $X^t$ maximizes the total value among all allocations in which each agent $i$ receives exactly $w_i$ items. Therefore, the left-hand side of the above expression, which is the difference between the sum of values in $B^t$ and the sum of values in $X^t$, must be at most $0$. But the right-hand side of the same expression is exactly the total cost of $C$. Therefore,
    $$ 0 \geq \mathbb{E}\left[\sum_{i\in N} \left( v_i(B^{t}_i) - v_i(X^{t}_i) \right)\right] = cost_{X^{t}}(C),$$
    so the cost of every directed cycle is at most $0$, as required.
\end{proof}
As the allocation in each iteration is WEF-able, the output allocation $X$ is WEF-able too.

To compute an upper bound on the subsidy, we adapt the proof technique in \cite{Brustle2020}.

\begin{lemma} \label{lemm:minimum cost subsidy}
    Let $X$ be a WEF-able allocation. 
    For any positive number $z$,
    if $cost_X(i,k) \geq -z$ for every edge $(i,k)$ in $G_{X,w}$, then the maximum subsidy required is at most $w_i z$ per agent $i \in N$.
\end{lemma}
\begin{proof} 
    Assume $P_i(X) = (i \ldots j)$ is the highest-cost path from $i$ in $G_{X,w}$. Note that $\ell_i(X) = cost_X(P_i(X))$.
    Then it holds for the cycle $C = (i \ldots j)$ that $cost_X(C) = \ell_i(X) + cost_X(j,i)$.
    By \Cref{th:ef}, $cost_X(C) \leq 0$, thus, $\ell_i(X) \leq -cost_X(j,i) \leq z.$
    Therefore, $p_i = w_i   \ell_i(X) \leq w_i z$.
\end{proof}

In most allocations, including the one resulting from \Cref{alg:general-additive}, it is not always true that $cost_X(i,k) \geq -V$ for every edge $(i,k)$. We use \Cref{lemm:minimum cost subsidy} with a modified valuation function $\Bar{v}_i$, derived from the weighted valuation $\frac{v_i(X_i)}{w_i}$.

We prove that an allocation that is WEF-able for the original valuations is also WEF-able for the modified valuations (\Cref{modified function wef-able}), and that the maximum subsidy required by each agent for the original valuations is bounded by the subsidy required for the modified valuations (\Cref{max sub modified is max sub original}).

Next, we demonstrate that under the modified valuations, the cost of each edge is at least $-V$. Finally, by \Cref{lemm:minimum cost subsidy}, we conclude that the maximum subsidy required for any agent $i\in N$ is $w_i   V$ for the modified valuation $\Bar{v}$ (\Cref{max subsidy for modified valuations}) and for the original valuations $v$ as well.

Let $X^t$ be the output allocation from \Cref{alg:general-additive}, computed in iteration $t$.
For each $i\in N$, we define the modified valuation function as follow:
\[
\Bar{v}_i(X_j^t) = \begin{cases}
    \frac{v_i(X_i^t)}{w_i} & j = i \\
    \frac{v_i(X_j^T)}{w_j} & j\neq i, t = T \\
    \max{\left(\frac{v_i(X_j^t)}{w_j}, \frac{v_i(X_i^{t+1})}{w_i}\right)} & j\neq i, t < T
\end{cases}
\]
Under the modified valuations, for any two agents $i,j \in N$, the modified-cost assigned to the edge $(i,j)$ in the envy graph (with unit weights) is defined as $\overline{cost}_{X}(i,j) = \Bar{v}_i(X_i) - \Bar{v}_i(X_j)$. Moreover, the modified-cost of a path $(i_1,...,i_k)$ is $$\overline{cost}_{X}(i_1,...,i_k) = \sum_{j=1}^{k-1} \overline{cost}_{X}(i_{j}, i_{j+1}).$$

\begin{observation} \label{observation:modified func}
For agent $i \in N$ and round $t\in [T]$ it holds that:
    \begin{enumerate}
        \item $
        \Bar{v}_i(X_i^t)
        =
        \frac{v_i(X_i^t)}{w_i}$.
        \item For agent $j\neq i \in N$, $
        \Bar{v}_i(X_j^t)\geq \frac{v_i(X_j^t)}{w_j}$;
        hence $\Bar{v}_i(X_j^t) - \Bar{v}_i(X_i^t)\geq \frac{v_i(X_j^t)}{w_j}-\frac{v_i(X_i^t)}{w_i}$.
    \end{enumerate}
\end{observation}

\begin{proposition} \label{modified function wef-able}
    Assume $X$ is WEF-able under the original valuations $v$. Then, $X$ is EF-able (i.e., WEF-able with unit weights) under the modified valuations $\Bar{v}$. 
\end{proposition}

\begin{proof} 
By \Cref{th:ef}, it is sufficient to prove that all directed cycles in the envy graph (with the modified valuations and unit weights) 
have non-positive total cost.

We prove a stronger claim: in every round $t$, 
the total modified-cost added to every directed cycle $C$ is non-positive. 

Let $X^t$ be the allocation computed by \Cref{alg:general-additive} at iteration $t$. 
Suppose, contrary to our assumption, that there exists a cycle $C=(i_1,...,i_r)$ and a round $t$
in which the modified-cost added to $C$ is positive. To simplify notation, we consider $i_1$ as $i_{r+1}$. 
This implies that
\begin{align} \label{contradiction: modified valuation}
    \sum_{j=1}^r \Bar{v}_{i_j}(X_{i_{j+1}}^t) > \sum_{j=1}^r \Bar{v}_{i_j}(X_{i_j}^t).
\end{align}
There are several cases to consider.

\underline{\bf Case 1:}  All arcs $i_j\to i_{j+1}$ in $C$ have $$\Bar{v}_{i_j}(X_{i_{j+1}}^t) = \frac{v_{i_j}(X_{i_{j+1}}^{t})}{w_{i_{j+1}}}$$
    (in particular, this holds for $t = T$).
    In this case, inequality \eqref{contradiction: modified valuation} implies 
    $$\overline{cost}_{X^t}(C) = \sum_{j=1}^r \frac{v_{i_j}(X_{i_j}^{t})}{w_{i_{j+1}}} - \frac{v_{i_j}(X_{i_j}^{t})}{w_{i_j}} > 0.$$ Combined with \Cref{th:ef}, this contradicts \Cref{general additive: WEF-able}, which states that $X^t$ is WEF-able.

\underline{\bf Case 2:}    
    All arcs $i_j\to i_{j+1}$ in $C$ have $$\Bar{v}_{i_j}(X_{i_{j+1}}^t) = \frac{v_{i_j}(X_{i_j}^{t+1})}{w_{i_j}}.$$
    In this case, inequality 
    \eqref{contradiction: modified valuation} implies 
    $$\sum_{j=1}^r \frac{v_{i_j}(X_{i_j}^{t+1})}{w_{i_j}} > \sum_{j=1}^r \frac{v_{i_j}(X_{i_j}^{t})}{w_{i_j}}.$$ 
    Notice that all the items in $X_{j_{1}}^{t+1},..., X_{j_{r}}^{t+1}$ are available at iteration $t$, which contradicts the optimality of $\{X_i^t\}_{i\in N}$.
    
\underline{\bf Case 3:} 
Some arcs $i_j\to i_{j+1}$ in $C$ satisfy Case 1 and the other arcs satisfy Case 2.

Let $l \geq 1$ be the number of arcs in $C$ that satisfy Case 2.
We decompose $C$ into a sequence of $l$ edge-disjoint paths, denoted $P_1, ... , P_l$, such that the last node of each path is the first node of the next path,
and in each path, only the last edge satisfies Case 2.

Formally, suppose that some path contains $k\geq 1$ agents, denoted as $i_1, ..., i_{k}$, and $k-1$ arcs.
Then  for each $1\leq j \leq k - 2$, $$\Bar{v}_{i_{j}}\left(X_{i_{j+1}}^t\right) = \frac{v_{i_{j}}\left(X_{i_{j+1}}^t\right)}{w_{i_{j+1}}},$$ and $$\Bar{v}_{i_{k-1}}\left(X_{i_{k}}^t\right) = \frac{v_{i_{k-1}}\left(X_{i_{k-1}}^{t+1}\right)}{w_{i_{k-1}}}.$$ 
    Since $\overline{cost}_{X^t}(C) > 0$, there exits a path $P = (i_1,..., i_k)$ where $\overline{cost}_{X^t}(P) > 0$, which implies that: 
    $$0 < \sum_{j = 1}^{k-1} \left( \Bar{v}_{i_{j}}(X_{i_{j+1}}^t) - \Bar{v}_{i_{j}}(X_{i_{j}}^t) \right) =
    \sum_{j=1}^{k-2} \left(\frac{v_{i_{j}}(X_{i_{j+1}}^t)}{w_{i_{j+1}}} - \frac{v_{i_{j}}(X_{i_{j}}^t)}{w_{i_{j}}} \right) + \frac{v_{i_{k-1}}(X_{i_{k-1}}^{t+1})}{w_{i_{k-1}}} - \frac{v_{i_{k-1}}(X_{i_{k-1}}^{t})}{w_{i_{k-1}}}.$$
    The rest of the proof is similar to the proof of \Cref{general additive: WEF-able}.
    We construct another allocation $B^t$ randomly as follows: 
    \begin{enumerate}
        \item For each agent $1\leq j \leq k-1$, we choose one item $o_{i_{j+1}}^t$ uniformly from $i_{j+1}$'s bundle and transfer it to $i_{j}$'s bundle.
        \item We choose one item $o_{i_{1}}^t$ uniformly from $i_1$'s bundle to remove.
        \item We choose one item $o_{i_{k-1}}^{t+1}$ uniformly from $X_{i_{k-1}}^{t+1}$ and add it to $i_{k-1}$'s bundle. 
    \end{enumerate}
    Thus, the expected change in value between $B^t$ and $X^t$ is
    \begin{align*}
        &\mathbb{E}\left[\sum_{i\in N} \left( v_i(B_i^t) - v_i(X^t_i) \right)\right] = \\ & \sum_{1\leq j\leq k-2} \left( \mathbb{E}\left[v_{i_j}(o_{i_{j+1}}^t)\right] - \mathbb{E}\left[v_{i_j}(o_{i_j}^t)\right] \right) + \mathbb{E}\left[v_{i_{k-1}}(o_{i_{k-1}}^{t+1})\right] - \mathbb{E}\left[v_{i_{k-1}}(o_{i_{k-1}}^t)\right] = \\
        &\sum_{1 \leq j \leq k-2} \frac{v_{i_j}(X^t_{i_{j+1}})}{w_{i_{j+1}}} - \frac{v_{i_j}(X^t_{i_j})}{w_{i_j}} + \frac{v_{i_{k-1}}(X_{i_{k-1}}^{t+1})}{w_{i_{k-1}}} - \frac{v_{i_{k-1}}(X_{i_{k-1}}^{t})}{w_{i_{k-1}}}.
    \end{align*}
    This is exactly the cost of $P$ which by assumption is greater than 0.
    However, according to \Cref{alg:general-additive}, $X^t$ maximizes the value of an allocation where each agent $i$ receives $w_i$ items among the set of $O^t$ items. Therefore,
     $$0 \leq \mathbb{E}\left[\sum_{i\in N} \left( v_i(B^{t}_i) - v_i(X^{t}_i) \right)\right] = cost_{X^{t}}(P)$$ 
    leading to a contradiction.

To sum up, $X^t$ is WEF-able under the original valuations $v$ (with weights $w$), and under the modified valuations $\Bar{v}$ (with unit weights).
\end{proof}
\begin{proposition} \label{max sub modified is max sub original}
    For the allocation $X$ computed by \Cref{alg:general-additive}, the subsidy required by an agent given $v$ (with weights $w$) is at most the subsidy required given $\Bar{v}$ (with unit weights).
\end{proposition}
\begin{proof}
    Given \Cref{observation:modified func}, for each $i,j \in N$,
    $$\Bar{v}_i(X_j) - \Bar{v}_i(X_i) \geq \frac{v_i(X_j)}{w_j} - \frac{v_i(X_i)}{w_i}.$$
    Thus, the cost of any path in the envy graph under the modified function and unit weights 
    is at least the cost of the same path in the weighted envy-graph with the original valuations. 
\end{proof}
\begin{proposition} \label{max subsidy for modified valuations}
    For the allocation $X$ computed by \Cref{alg:general-additive}, the subsidy to each agent is at most $w_i   V$ for the modified valuation profile $\Bar{v}$.
\end{proposition}
\begin{proof}
    By \Cref{modified function wef-able}, the allocation $X$ is WEF-able under the valuations $\Bar{v}$. 
    Together with \Cref{lemm:minimum cost subsidy}, if for each $i,j\in N$ it holds that $\Bar{v}_i(X_j) - \Bar{v}_i(X_i) \geq -V$, the subsidy required for agent $i\in N$ is at most $w_i   V$ for $\Bar{v}$.

    \begin{align*}
        &\Bar{v}_i(X_j) - \Bar{v}_i(X_i) = \sum_{t\in [T]}\Bar{v}_i(X_j^t) - \sum_{t\in [T]}\Bar{v}_i(X_i^t) =  \\ &\sum_{t\in [T-1]}\max\Big\{\frac{v_i(X_j^t)}{w_j}, \frac{v_i(X_i^{t+1})}{w_i}\Big\} + \frac{v_i(X_j^T)}{w_j}- \sum_{t\in [T]}\frac{v_i(X_i^t)}{w_i} \geq \\ &
         \sum_{t\in [T-1]}\frac{v_i(X_i^{t+1})}{w_i}+ \frac{v_i(X_j^T)}{w_j}- \sum_{t\in [T]}\frac{v_i(X_i^t)}{w_i} = \\ &\frac{v_i(X_j^T)}{w_j} - \frac{v_i(X_i^1)}{w_i} \geq - \frac{v_i(X_i^1)}{w_i}.
    \end{align*}
Since $X_i^1$ contains exactly $w_i$ items, 
$-v_i(X_i^1) \geq - w_i  V$.
Hence, $\Bar{v}_i(X_j) - \Bar{v}_i(X_i) \geq -\frac{w_i   V}{w_i} = -V$.
\end{proof}

We are now prepared to prove the main theorem.
\begin{theorem} \label{theorem: sub general additive}
    For additive valuations and integer entitlements, \Cref{alg:general-additive} computes in polynomial time a WEF-able allocation where the subsidy to each agent is at most $w_i V$ and the total subsidy is at most 
    $(W-\wmin)V$.
\end{theorem}
\begin{proof}
For the runtime analysis, the most computationally intensive step in \Cref{alg:general-additive} is solving the maximum integral flow of minimum cost in $G'$. The flow network $G'$ consists of at most $n+m+2$ nodes and at most $n+m+mn$ arcs. By \cite{goldberg1989finding}, this can be done in time polynomial in $n,m$:
\\
\resizebox{\textwidth}{!}{$ O\left( \left(n+m+2\right) \left( n+m+mn \right) \log\left(n+m+2\right) \min\{ \log \left(\left(n+m+2\right) V \right) , \left( n+m+mn \right) \log \left(n+m+2\right) \}\right). $}


By \Cref{general additive: WEF-able}, $X$ is WEF-able under the original valuations. Combined with \Cref{modified function wef-able} and \Cref{max subsidy for modified valuations} , $X$ is also WEF-able under the modified valuations and requires a subsidy of at most $w_i V$ for each agent $i\in N$. 

\Cref{max sub modified is max sub original}, implies that under the original valuations, the required subsidy for each agent $i\in N$ is at most $w_i V$. 
By \Cref{thm:minsubsidy}, there is at least one agent who requires no subsidy, so the required total subsidy is at most $(W-\wmin)  V$.
\end{proof}

The WEF condition is invariant to multiplying the weight vector by a scalar.
This can be used in two ways:

(1) If the weights are not integers, but their ratios are integers, we can still use Algorithm~\ref{alg:general-additive}.
 For example, if $w_1=1/3$ and $w_2=2/3$ (or even if $w_i$'s are irrational numbers such as $w_1=\sqrt{2}$ and $w_2=2\sqrt{2}$), Algorithm 1 works correctly by resetting $w_1=1$ and $w_2=2$.

(2) If the weights are integers with greatest common divisor (gcd) larger than 1, we can divide all weights by the gcd to get a better subsidy bound.


\begin{lemma} \label{cor: sub general additive}
    For additive valuations and integer entitlements, there exists an algorithm that computes in polynomial time a WEF-able allocation where the subsidy to each agent is at most $w_i V/\gcdw$ and the total subsidy is at most $(W-\wmin)V/\gcdw$, where $\gcdw$ is the greatest common divisor of all the $w_i$.
\end{lemma}
\begin{proof}
Algorithm~\ref{alg:general-additive} works correctly, even if we divide each $w_i$ by the greatest common divisor of $w_i$'s.
In other words, letting $d={\rm gcd}(w_1,...,w_n)$, $w'_i=w_i/d$, $W'=W/d$, and running Algorithm 1 with $w'_i$'s, we get the bound $(W'-w'_{\min})V$ of the total subsidy.
\end{proof}

A discussion about the tightness of the bound can be found in Appendix \ref{alg:general-additive tightness}.

\subsection{Combining WEF-able and Approximate-WEF}
\label{sec:app_EF}
In the setting without money and additive valuations, WEF can be relaxed by allowing envy up to an upper bound based on item values. We adopt the generalization of allowable envy, \emph{WEF$(x,y)$}, proposed by \cite{chakraborty2022weighted}
\footnote{
The definition of WEF$(x,y)$ does not apply to non-additive valuations. \cite{montanari2024weighted} introduced two extensions to this definition; but they are outside the scope of our paper.
}, as well as another relaxation of WEF, \emph{WWEF1}, introduced in \cite{CISZ21}. 
\begin{definition}[\cite{chakraborty2022weighted}] For $x,y \in [0,1]$, an allocation $X$ is said to satisfy \emph{WEF(x,y)} if for any $i,j \in N$, there exists $B\subseteq X_j$ with $|B| \leq 1$ such that $\frac{v_i(X_i) + y  v_i(B)}{w_i} \geq \frac{v_i(X_j) - x  v_i(B)}{w_j}$.
\end{definition}
WEF$(x,y)$ captures various conditions related to WEF: 
\begin{itemize}
    \item 
 WEF$(0,0)$ corresponds to WEF.
    \item 
    WEF$(1,0)$ coincides with (strong) weighted envy-freeness up to one item (WEF1) (\cite{CISZ21}).
    \item 
    WEF$(1,1)$ coincides with weighted envy-free up to one item transfer (WEF1-T) (\cite{AGM23a,HSV23a}).
\end{itemize}

\begin{definition}[\cite{CISZ21}]
An allocation $X$ is said to be \emph{weakly weighted envy-free up to one item} (WWEF1) if for any pair of agents $i,j$ with $X_j\neq\emptyset$, there exists an item $o \in X_j$ such that either
$ \frac{v_i(X_i)}{w_i} \ge  \frac{v_i(X_j \setminus \{o\})}{w_j} $ or $\frac{v_i(X_i\cup\{o\})}{w_i} \ge  \frac{v_i(X_j)}{w_j}$.
\end{definition}

Halpern and Shah proved in \cite{HaSh19a} that, if an allocation $X$ is EF-able and EF1, 
the total subsidy of at most $(n-1)^2   V$ is sufficient.
The following theorem generalizes this result to the weighted setting.
\begin{theorem}
    \label{WEF-able and WEF(x,y) subsidy bound2}
   Let $X$ be both WEF-able and WEF$(x,y)$ for some $x,y \in [0,1]$. Then there exists a subsidy vector $\mathbf{p}$, such that $(X,\mathbf{p})$ WEF, with total subsidy at most $(x+y)\left(\frac{W}{\wmin} - 1 \right) (n-1)V$.
\end{theorem}
\begin{proof}
    $X$ is WEF$(x,y)$, so for all $i,j\in N$, there exists $B\subseteq X_j$ with $|B|\leq 1$ where
    \begin{align*}
        & \frac{v_i(X_j)}{w_j} - \frac{v_i(X_i)}{w_i} \leq \frac{y  v_i(B)}{w_i} + \frac{x  v_i(B)}{w_j} \le \frac{y  v_i(B)}{\wmin} + \frac{x  v_i(B)}{\wmin} = 
         \frac{(x+y)  v_i(B)}{\wmin} \le \\ &(x+y)\frac{V}{\wmin}.
    \end{align*}
    Any path contains at most $n-1$ arcs, that is, $p_i = w_i \ell_i \leq (x+y)\frac{w_i}{\wmin}   (n-1)  V$.
    
    By \Cref{thm:minsubsidy}, there is at least one agent that requires no subsidy, so the required total subsidy is at most $(x+y)\frac{W - \wmin}{\wmin} (n-1)  V = (x+y)\left(\frac{W}{\wmin} - 1 \right) (n-1)V$.
\end{proof}

\cite{Brustle2020} proved that the allocation resulting from their algorithm satisfies both EF and EF1. However, in the weighted setup, a WEF-able allocation may not satisfy WEF$(x,y)$ for any $x,y$ with $x+y<2$.
    This also holds for the allocation produced by \Cref{alg:general-additive}.

\begin{proposition} \label{app:general additive wef2}
For any $x, y\geq 0$ with $x+y<2$,
there exists a weight vector and an instance with additive valuations in which that every WEF-able allocation fails to satisfy WEF$(x, y)$ or WWEF1.
\end{proposition}

\begin{proof}
Consider an instance with two identical items and two agents with weights $w_1<w_2$.
Agent 1 values each item at $1$ and agent 2 values each item at $2$.

The only WEF-able allocation is the one giving both items to agent 2: $X_1 = \emptyset, X_2 = \{o_1,o_2\}$. (If we allocate one item to each agent, the cost of the cycle between those agents would be $\left(\frac{1}{w_2} - \frac{1}{w_1}\right)\left(v_2(o) - v_1(o)\right)$, which is positive).

For this allocation to be WEF$(x,y)$, it must satisfy 
\begin{align*}
&
\frac{v_1(X_1) + y  v_1(o_1)}{w_1} \geq \frac{v_1(X_2) - x  v_1(o_1)}{w_2}    
\iff 
\frac{y}{w_1} \geq \frac{2 - x}{w_2}
\iff
x + \frac{w_2}{w_1} y \geq 2.
\end{align*}
Hence, if $x+y<2$ and $\frac{w_2}{w_1}$ is sufficiently close to 1, then $X$ fails WEF$(x,y)$.

Moreover, $X$ does not satisfy WWEF1, as $0 = \frac{v_1(X_1)}{w_1} < \frac{v_1(X_2 \setminus \{o\})}{w_2} = 1$ and $\frac{1}{w_1} =\frac{v_1(X_1 \cup \{o\})}{w_1} < \frac{v_1(X_2)}{w_2} = \frac{2}{w_2}$, for any $o\in X_2$, whenever $w_2 < 2 w_1$.
\end{proof}

We also observe that even for 2 agents, the \emph{Weighted Picking Sequence Protocol} (\cite{CISZ21}), which outputs a WEF$(1,0)$ allocation for any number of agents with additive valuations, is not WEF-able: 


\begin{example}
\label{exm:picking-sequence-not-wefable}
Suppose there is one item $o$ and $v_1(o)=2$, $v_2(o)=1$, $w_1=1$, and $w_2=4$. 
Agent 2 gets the first turn and gets $o$. 
Agent 2 gets the first turn and gets $o$. 
The envy of agent 1 toward agent 2 is $2/4$, while the envy of agent 2 toward agent 1 is $-1/4$. Thus, the cycle between these agents has a positive cost: $2/4-1/4 = 1/4 > 0$. By \Cref{th:ef}, the resulting allocation is not WEF-able.
\end{example}

\paragraph{\textbf{Biased Weighted Adjusted Winner Procedure}}
\Cref{app:general additive wef2} is partly complemented by the result below (Theorem~\ref{thm:WEFability-WEF1T-additive-two-agents}) that states that WEF-ability and WEF(1,1) are compatible for two agents having additive valuations. 
The theorem uses a particular version of the \emph{Weighted Adjusted Winner} procedure (\cite{CISZ21}). The original procedure finds a WEF$(1,0)$ and Pareto efficient allocation.
We call our variant \emph{Biased Weighted Adjusted Winner Procedure}, as it is biased towards the agent who expresses a higher value for a `contested' item. The resulting allocation may not be WEF$(1,0)$, but it is WEF$(1,1)$ and WEF-able.

We first observe that for two agents, the WEF condition for each agent $i$ is equivalent to the \emph{weighted proportionality} (WPROP) condition:
$v_i(X_i) \geq \frac{w_i}{W}\cdot v_i(M)$.

\begin{enumerate}
\item Normalize the valuations so that the sum of values over all items is $1$ for both agents. 
Sort the items such that $\frac{v_1(o_1)}{v_2(o_1)} \ge \frac{v_1(o_2)}{v_2(o_2)} \ge \dots \ge \frac{v_1(o_m)}{v_2(o_m)}$.
\item 
Let $d \in \{1,2,...,m\}$ be the unique number satisfying $\frac{1}{w_1}\sum_{r=1}^{d-1}v_1(o_r) < \frac{1}{w_2}\sum_{r=d}^{m}v_1(o_r)$ and $\frac{1}{w_1}\sum_{r=1}^{d}v_1(o_r) \ge \frac{1}{w_2}\sum_{r=d+1}^{m}v_1(o_r)$. We call $o_d$ the \emph{contested object}.
\item 
Denote $L := \{o_1,\ldots, o_{d-1}\}$ and $R := \{o_{d+1},\ldots, o_{m}\}$ (the ``Left" and ``Right" parts); note that each of them might be empty.
Give $L$ to agent $1$ and $R$ to agent $2$.
\item 
Finally, give $o_d$ to the agent $i$ with largest $v_i(o_d)$ (break ties arbitrarily).
\end{enumerate}

\begin{theorem}\label{thm:WEFability-WEF1T-additive-two-agents}
The outcome of the \emph{Biased Weighted Adjusted Winner Procedure} is both WEF$(1,1)$ and WEF-able.
\end{theorem}
\begin{proof}
The WEF$(1,1)$ for agent 1 follows immediately from the definition of $o_d$, as 
\begin{align*}
\frac{1}{w_1}v_1(L\cup \{o_d\})\geq \frac{1}{w_2}v_1(R).    
\end{align*}
If $o_d$ is allocated to 1 then 1 does not envy at all,  otherwise 1 stops envying after $o_d$ is transferred to him.

As for agent 2, the item ordering implies that
\begin{align*}
&
\frac{v_1(L)}{v_2(L)} \geq \frac{v_1(R\cup \{o_d\})}{v_2(R\cup \{o_d\})}
\iff
\frac{v_2(R\cup \{o_d\})}{v_2(L)} \geq \frac{v_1(R\cup \{o_d\})}{v_1(L)}
\end{align*}
By the definition of $o_d$, 
$\frac{v_1(L)}{v_1(R\cup \{o_d\})} < \frac{w_1}{w_2}$.
Combining this with the above inequality implies
\begin{align*}
\frac{v_2(R\cup \{o_d\})}{v_2(L)} > \frac{w_2}{w_1}
\iff
\frac{1}{w_2}v_2(R\cup \{o_d\}) > \frac{1}{w_1}v_2(L).
\end{align*}
An analogous argument to that used for agent 1 shows that the WEF$(1,1)$ condition is satisfied for 2 too.

Next, we prove that the outcome is WEF-able. 

Suppose $o_d$ is given to agent 2, that is, $v_2(o_d)\geq v_1(o_d)$. By definition of $o_d$, there exists a fraction $r\in [0,2]$ such that
\begin{align*}
\frac{1}{w_1}(v_1(L) + r\cdot v_1(o_d)) = \frac{1}{w_2}(v_1(R\cup \{o_d\}).
\end{align*}
Let the subsidy to agent 1 be $p_1 := r\cdot v_1(o_d)$, and the subsidy to agent 2 be $p_2 := 0$. The resulting allocation is WEF for agent 1 by definition. 

As for agent 2, the item ordering implies 
\begin{align*}
\frac{v_1(L) + r v_1(o_d)}{v_2(L)+ r v_2(o_d)} \geq \frac{v_1(R\cup \{o_d\})}{v_2(R\cup \{o_d\})} 
\iff
v_2(L)+r v_2(o_d) &\leq \frac{v_2(R\cup \{o_d\})}{v_1(R\cup \{o_d\})} (v_1(L)+r v_1(o_d))
\\
& = \frac{w_1}{w_2}v_2(R\cup \{o_d\}).
\end{align*}
As $v_1(o_d)\leq v_2(o_d)$, agent 2 values agent 1's bundle at
\begin{align*}
v_2(L) + p_1 
=
v_2(L) + r \cdot v_1(o_d)
\leq v_2(L) + r\cdot v_2(o_d),
\end{align*}
which by the above inequality is at most $\frac{w_1}{w_2}v_2(R\cup \{o_d\})$. Hence, the allocation is WEF for agent 2 too.

The case that $o_d$ is allocated to agent 1 can be proved in a similar way. 
\end{proof}

It remains open whether weighted envy-freeability and WEF(1,1) are compatible for $n\geq 3$ agents. 

\section{Identical Additive Valuation}
\label{sec:identical_additive_val}
This section deals with the case where all agents have identical valuations, that is, $v_i\equiv v$ for all $i\in N$.
With identical valuations, any allocation is WEF-able by \Cref{lem:max-value}.
Furthermore, all allocations are non-wasteful. 
We present a polynomial-time algorithm for finding a WEF-able allocation with a subsidy bounded by $V$ per agent and a total subsidy bounded by $(n-1)V$. 
The following shows that this bound is tight for any weight vector:
\begin{theorem}\label{thm:subsidy-lb-identical-additive}
For identical additive valuations, 
for any integer weights vector, there exists an instance where, for any WEF-able allocation, at least one agent requires subsidy at least $V$, and the total subsidy is at least $(n-1)V$.
\end{theorem}
\begin{proof}
    Consider $n$ agents with integer weights $w_1 \leq \cdots \leq w_n$ and $m = 1 + \sum_{i\in N} \left( w_i - 1 \right)$ items all valued at $V$.
    
    To avoid envy, each agent $i$ should receive a total utility of $w_i V$, so the sum of all agents' utilities would be $W V$. 
    
    As the sum of all values is $(W-(n-1))V$, a total subsidy of at least $(n-1)V$ is required
    (to minimize the subsidy per agent, each agent $i\in N$ should receive $w_i - 1$ items, except for the agent with the highest entitlement (agent $n$), who should receive $w_n$ items.
    
    The value per unit entitlement of each agent $i<n$ is $V(w_i-1)/w_i$, 
    and for agent $n$ it is $V$.
    Therefore, to avoid envy, each agent $i<n$ should receive a subsidy of $w_i \left(1 - \frac{w_i - 1}{w_i}\right)V = V$ and the total subsidy required is $(n-1)V$.
\end{proof}
\Cref{alg:one-subsidy-identical-additive} presents our WEF-able allocation method with bounded subsidy.

The algorithm traverses the items in an arbitrary order. At each iteration it selects the agent that minimizes the expression $\frac{v(X_i\cup{\{o\}})}{w_i}$, with ties broken in favor of the agent with the larger $w_i$, and allocates the next item to that agent.
Intuitively, this selection minimizes the likelihood that weighted envy is generated.

\begin{algorithm}[h]
\caption{Weighted Sequence Protocol For Additive Identical Valuations}\label{alg:one-subsidy-identical-additive}
\KwIn{Instance $(N,M,v, \mathbf{w})$ with additive identical valuations.}
\KwOut{WEF-able allocation $X$ with total required subsidy of at most $(n-1)  V$.}
$X_{i} \gets \emptyset, \forall i\in N$\;
\For{$o: 1$ to $m$}{
 $U \gets \arg\min_{i\in N}\frac{v\left(X_i \cup \{o\}\right)}{w_i}$\;
$u \gets \max_{i\in U}\left(i\right)$\;
 Add $o$ to $X_u$
}
\Return $X$
\end{algorithm}

The following example illustrates \Cref{alg:one-subsidy-identical-additive}:
\begin{example}\label{example:additive identical}
Consider two agents, denoted as $i_1$ and $i_2$, with corresponding weights $w_1 = 1$ and $w_2 = \frac{7}{2}$, and three items, namely $o_1, o_2, o_3$, with valuations $v(o_1) = v(o_2) = v(o_3) = 1$, \Cref{alg:one-subsidy-identical-additive} is executed as follows:
\begin{enumerate}
\item for the first iteration, the algorithm compares $\frac{v(o_{1})}{w_1} = 1 $ and $\frac{v(o_{1})}{w_2} = \frac{2}{7}$. Consequently, the algorithm allocates item $o_1$ to agent $i_2$, resulting in $X_1 = \emptyset$ and $X_2 = \{o_1\}$.
\item for the second iteration, the algorithm compares $\frac{v(o_{2})}{w_1} = 1$ and $ \frac{v(X_2 \cup \{o_{2}\})}{w_2} = \frac{2}{\frac{7}{2}} = \frac{4}{7}$. Subsequently, the algorithm allocates item $o_2$ to agent $i_2$, resulting in $X_1 = \emptyset$ and $X_2 = \{o_1, o_2\}$.
\item for the third iteration, the algorithm compares $\frac{v(o_3)}{w_1} = 1$ and $\frac{v(X_2 \cup \{o_3\})}{w_2} = \frac{3}{\frac{7}{2}} = \frac{6}{7}$, Consequently, item $o_3$ is allocated to agent $i_2$, resulting in $X_1 = \emptyset$ and $X_2 = \{o_1, o_2, o_3\}$.
\end{enumerate}
Now, agent $i_1$ envies agent $i_2$ by an amount of $\frac{v(X_2)}{w_2}-\frac{v(X_1)}{w_1} = \frac{3}{\frac{7}{2}} = \frac{6}{7}$, and conversely, agent $i_2$ envies agent $i_1$ by $\frac{v(X_1)}{w_1}-\frac{v(X_2)}{w_2} = -\frac{6}{7}$
In order to mitigate envy, $p_{1} = \frac{6}{7}$ and $p_{2} = 0$.
\end{example}

We start by observing that, with identical valuations, the cost of any path in the weighted envy graph is determined only by the agents at the endpoints of that path.
\begin{observation}
\label{cost_of_path_identical_valuatoins}
Given an instance with identical valuations,
let $X$ be any allocation, and denote by $P$ any path in the weighted envy-graph of $X$ between agents $i,j\in N$.
\begin{align*}
    cost_X(P) = \frac{v(X_j)}{w_j}-\frac{v(X_i)}{w_i}
\end{align*}
This is because the path cost is 
$
    \sum_{(h,k)\in P} cost_X(h,k)= \sum_{(h,k)\in P}\frac{v(X_k)}{w_k}-\frac{v(X_h)}{w_h},
$ 
and the latter sum is a telescopic sum that reduces to the difference of its last and first  element.
\end{observation}

\Cref{example:additive identical} illustrates that the resulting allocation may not be \emph{WEF(1,0)} ---
\Cref{alg:one-subsidy-identical-additive} might allocate all items to the agent with the highest entitlement.
However, the outcome is always WEF$(0,1)$:
\begin{proposition} \label{identical additive wef01}
    For additive identical valuations, \Cref{alg:one-subsidy-identical-additive} computes a WEF$(0,1)$ allocation.
\end{proposition}
\begin{proof}
We prove by induction that at each iteration, $X$ satisfies WEF$(0,1)$. 
The claim is straightforward for the first iteration. Assume the claim holds for the $(t-1)$-th iteration, and prove it for the $t$-th iteration. Let $o$ be the item assigned in this iteration and $u$ be the agent receiving this item.
Agent $u$ satisfies WEF$(0,1)$ due to the induction hypothesis.
For $i \neq u$, by the selection rule, $\frac{v(X_u)}{w_u} \leq \frac{v(X_i \cup \{o\})}{w_i}$.
This is exactly the definition of WEF$(0,1)$.
\end{proof}
\begin{proposition}
\label{prop:wef01 upper bound}
With identical additive valuations, for every WEF$(0,1)$ allocation $X$,
$\ell_i(X) \leq \frac{V}{w_i}$, for all $i\in N$.
\end{proposition}
\begin{proof} 
For each agent $i\in N$, 
denote the highest-cost path starting at $i$ in that graph by $P_{i}(X) = (i, ..., j)$ for some agent $j\in N$. 
Then by \Cref{cost_of_path_identical_valuatoins},
$\ell_i(X) = cost_{X}(P_i(X)) = \frac{v(X_{j})}{w_{j}} - \frac{v(X_{i})}{w_i}$.

From the definition of WEF$(0,1)$, \Cref{identical additive wef01} implies that this difference is at most $\frac{v(o)}{w_i}$ for some object $o\in X_j$. 
Therefore, the difference is at most $\frac{V}{w_i}$.
\end{proof}

\begin{theorem}\label{thm:poly-on-subsidy-identical-additive}
For identical additive valuation, there exists a polynomial time algorithm to find a WEF-able and non-wasteful allocation such that the subsidy per agent is at most $V$.
Therefore, the total subsidy required is at most $(n-1)  V$.
\end{theorem}
\begin{proof}
It is clear that \Cref{alg:one-subsidy-identical-additive} runs in polynomial time. 

Let $X$ be the output of \Cref{alg:one-subsidy-identical-additive}.
First notice that $X$ is WEF-able and non-wasteful. 

    \Cref{prop:wef01 upper bound} implies that, to achieve weighted-envy-freeness under identical additive valuations for the allocation computed by \Cref{alg:one-subsidy-identical-additive}, the required subsidy per agent $i \in N$ is at most $ w_i   \frac{V}{w_i} = V$. 
    In combination with \Cref{thm:minsubsidy}, the total required subsidy is at most $(n-1)  V$. 
\end{proof}


Note that $W \geq n   \wmin$. Therefore, this bound is better than the one proved in 
    \Cref{theorem: sub general additive} for integer weights: $\left(W - w_i \right)   V \geq \left(n - 1 \right) w_i  V \geq \left(n - 1 \right)   V $.

The upper bound of $(n-1)V$ is tight even for equal entitlements (\cite{HaSh19a}).
Interestingly, when either the valuations or the entitlements are identical, the worst-case upper bound depends on $n$, whereas when both valuations and entitlements are different, the bound depends on $W$.
\section{Binary Additive Valuation} \label{sec:binary_additive_val}
In this section we focus on the special case of agents with binary additive valuations.
We assume $v_i(o) \in \{0,1\}$ for all $i \in N$ and $o \in M$.
%
We start with a lower bound.
\begin{proposition}
\label{prop:lower-bound-binary}
For every $n\geq 2$ and weight vector $\mathbf{w}$, there is an instance with $n$ agents with binary valuations in which the total subsidy in any WEF allocation is at least $\frac{W}{w_2} - 1$.
\end{proposition}
\begin{proof}
    There is one item. 
Agents $1$ and $2$ value the item at $1$ and the others at $0$.
If agent $i \in \{1,2\}$ gets the item,
then the other agent $j\neq i \in \{1,2\}$ must get subsidy $\frac{w_j}{w_i}$.
To ensure that other agents do not envy $j$'s subsidy, every other agent $k \not\in \{1,2\}$ 
must get subsidy $\frac{w_k}{w_i}$.
The total subsidy is $\frac{W}{w_i} - 1$.
The subsidy is minimized by giving the item to agent $2$, since $w_2\geq w_1$.
This gives a lower bound of $\frac{W}{w_2} - 1$.
\end{proof}

Below, we show how to compute a WEF-able allocation where the subsidy given to each agent $i\in N$ is at most $\frac{w_i}{\wmin} V = \frac{w_i}{\wmin}$, and the total subsidy is at most $\frac{W}{\wmin} - 1$.

In the case of binary valuations, \Cref{alg:general-additive} is inefficient in three ways: (1) the maximum-value matching does not always prioritize agents with higher entitlements, (2) there may be situations where an agent prefers items already allocated in previous iterations, while the agent holding those items could instead take unallocated ones, and (3) it works only for agents with integer weights.

We address these issues by adapting the \emph{General Yankee Swap (GYS)} algorithm introduced by \cite{viswanathan2023general}.

GYS starts with an empty allocation for all agents.
We add a dummy agent $i_0$ and assume that all items are initially assigned to $i_0$:
$X_{i_0} = M$.

\Cref{alg:binary-additive} presents our approach for finding a WEF-able allocation with a bounded subsidy.
The algorithm runs in $T$ iterations. 
We denote by $X^t$ the allocation at the end of iteration $t\in [T]$.
Throughout this algorithm, we 
divide the agents into two sets:
\begin{enumerate}
    \item $R$: The agents remaining in the game at the beginning of the iteration $t$.
    \item $N \setminus R$: The agents who were removed from the game in earlier iteration $t' < t$. Agents are removed from the game when the algorithm deduces that their utility cannot be improved.
\end{enumerate}

As long as not all the objects have been allocated, at every iteration $t$, the algorithm looks for the agents maximizing the \emph{gain function} ( \cite{viswanathan2023general}) among $R$, i.e., the agents remaining in the game at this iteration.

We use the gain function: 
$\frac{w_i}{v_i(X_i^{t-1}) + 1}$, 
which selects agents with the minimal potential for increasing envy.
If multiple agents have the same value, we select one arbitrarily.

The selected agent then chooses either to acquire an unallocated item or take an item from another agent. In either case, their utility increases by 1. If the agent takes an item from another, the affected agent must decide whether to take an unallocated item or another allocated item to preserve their utility, and so on. This process creates a \emph{transfer path} from agent $i$ to the dummy agent $i_0$ , where items are passed until an unallocated item is reached.
%
Formally, we represent this as a directed graph, where nodes are agents, and an edge $(i,j)$ if and only if there exists an item in $j$'s bundle that $i$ values positively.
A \emph{transfer path} is any directed path in that graph, that ends at the dummy agent $i_0$.



When an agent is selected, the algorithm attempts to find a transfer path from that agent, preserving utilities for all agents except the initiator, whose utility increases by 1. If no path is found, the agent is removed from the game. We use the polynomial-time method by \cite{viswanathan2023general} to find transfer paths.

\Cref{alg:binary-additive} differs from $GYS$ in the following way: at the beginning of iteration $t$, the algorithm first removes all agents without a transfer path originating from them (line \ref{line:R}). Then, it selects an agent based on the gain function to allocate a new item to that agent.
For convenience, we denote by $R(t)$ the agents who have a transfer path originating from them at the beginning of iteration $t$ (line \ref{line:R}).

\begin{algorithm}[h]
\caption{
\label{alg:binary-additive}
Weighted Sequence Protocol For Additive Binary Valuations}
\KwIn{ Instance $(N,M,v, \mathbf{w})$ with additive binary valuations.}
\KwOut{ WEF-able allocation $X$ with total required subsidy of at most $\frac{W}{\wmin} - 1$.}
$X_{i_{0}} \gets M$, and $X_{i}^0 \gets \emptyset$ for each $i\in N$ \tcc{All items initially are unassigned}
$t \gets 1$\;
$R \gets N$\;
\While{$R \neq \emptyset$}{
Remove from $R$ all agents who do not have a transfer path starting from them \label{line:R}\;
 $u \gets \arg \max_{i \in R} \left( \frac{w_i}{v_i(X_i^{t-1}) + 1}\right)$ \tcc{\footnotesize Choose the agent who maximizes the gain function \normalsize}
Find a transfer path starting at $u$ 
    \tcc{\footnotesize For example, one can use the BFS algorithm to find a shortest path from $u$ to $i_0$. \normalsize}
     Transfer the items along the path and update the allocation $X^{t}$\;
     $t\gets t+1$
}
\Return $X^t$
\end{algorithm}


The following example demonstrates \Cref{alg:binary-additive}:
\begin{example} \label{app:example: binary additive}
Consider two agents with weights $w_1 = 1$ and $w_2 = 2$, and five items. The valuation functions are:
\[
\begin{bmatrix}
      & o_1 & o_2 & o_3 & o_4 & o_5\\
  i_1 & 1   & 1   & 1   & 1   & 1\\
  i_2 & 1   & 1   & 1   & 1   & 0
\end{bmatrix}
\]
The algorithm is executed as follows:

\begin{enumerate}
  \item For $t=1$, the algorithm compares $\frac{1}{w_1} = \frac{1}{1}$, $\frac{1}{w_2} = \frac{1}{2}$. Consequently, the algorithm searches for a transfer path starting at $i_2$ and ending at $i_{\text{0}}$, and finds the path $(i_2, i_{\text{0}})$. The algorithm transfers the item $o_1$ to agent $i_2$ from $i_{\text{0}}$'s bundle, resulting in $X_1^1 = \emptyset$ and $X_2^1 = \{o_1\}$.
  \item For $t=2$, the algorithm compares $\frac{1}{w_1} = \frac{1}{1}$ and $\frac{v_2(X_2^1) + 1}{w_2} = \frac{2}{2}$. Since those values are equal, the algorithm arbitrarily selects agent $i_2$ and searches
  for a transfer path starting at $i_2$ and ending at $i_{\text{0}}$, and finds the path $(i_2, i_{\text{0}})$. The algorithm transfers the item $o_2$ to agent $i_2$, yielding $X_1^2 = \emptyset$ and $X_2^2 = \{o_1, o_2\}$.
  \item For $t=3$, the algorithm compares $\frac{1}{w_1} = 1$ and $\frac{v_2(X_2^2) + 1}{w_2} = \frac{3}{2}$. As a result, the algorithm searches for a transfer path starting at $i_1$ and ending at $i_{\text{0}}$, and finds the path $(i_1, i_{\text{0}})$. The algorithm transfers the item $o_3$ to agent $i_1$, producing $X_1^3 = \{o_3\}$ and $X_2^3 = \{o_1, o_2\}$.
  \item For $t=4$, the algorithm compares $\frac{v_1(X_1^3) +1}{w_1} = 2$ and $\frac{v_2(X_2^3) + 1}{w_2} = \frac{3}{2}$. Thus, the algorithm searches
   for a transfer path starting at $i_2$ and ending at $i_{\text{0}}$, and finds the path $(i_2, i_{\text{0}})$. The algorithm transfers the item $o_4$ to agent $i_2$, leading $X_1^4 = \{o_3\}$ and $X_2^4 = \{o_1, o_2, o_4\}$.
  \item For $t=5$, the algorithm compares $\frac{v_1(X_1^4) +1}{w_1} = 2$ and $\frac{v_2(X_2^4) + 1}{w_2} = \frac{4}{2} = 2$. Since those values are equal, the algorithm arbitrarily selects agent $i_2$ and searches
  for a transfer path starting at $i_2$ and ending at $i_{\text{0}}$, and finds the path $(i_2, i_1, i_{\text{0}})$. The algorithm transfers the item $o_3$ to agent $i_2$ from $i_1$'s bundle and the item $o_5$ to agent $i_1$ from $i_{\text{0}}$'s bundle, leading $X_1^5 = \{o_5\}$ and $X_2^5 = \{o_1, o_2, o_3, o_4\}$.
  \item Agent 1 envies agent 2 by $\frac{4}{2} - \frac{1}{1} = 1$, while agent 2 envies agent 1 by $0 - \frac{4}{2} < 0$.
  \item In order to mitigate envy, $p_{1} = 1$ and $p_{2} = 0$.
\end{enumerate}
\end{example}

\begin{definition} [\cite{viswanathan2023general}]
    An allocation $X$ is said to be \textit{non-redundant} if for all $i\in N$, we have $v_{i}(X_i) = |X_i|$.
\end{definition}
That is, $v_j(X_i) \leq |X_i| = v_i(X_i)$ for every $i,j \in N$, and therefore, every non-redundant allocation is also WEF-able 
by \Cref{lem:max-value}.
%
Lemma 3.1 in \cite{viswanathan2023general} shows that the allocation produced by GYS is non-redundant. The same is true for our variant:
\begin{lemma}
\label{app:thm:wefable--iff-no-cycles9} 
At the end of any iteration $t$ of 
\Cref{alg:binary-additive},
the allocation $X^t$ is non-redundant. 
\end{lemma}
\begin{proof}
We prove by induction that at the end of each iteration $t$, $X^t$ remains non-redundant.

For the base case, $X^0$ is an empty allocation and is therefore non-redundant. 
Now, assume that at the end of iteration $t-1$, $X^{t-1}$ is non-redundant. 

If $X^{t} = X^{t-1}$, meaning no agent received a new item, the process is complete. Otherwise, let $u$ be the agent who receives new item. Agent $u$ obtains an item via the transfer path $P = (u = i_1, \ldots, i_k)$. For each $1 \leq j < k$, agent $i_j$ receives the item $o_j$ from the bundle of $i_{j+1}$, given that $v_{i_j}(o_j) = 1$. Agent $i_k$ receives a new item $o_k$ from the bundle of $i_0$, with $v_{i_k}(o_k) = 1$.

Additionally, for each $1 < j \leq k$, item $o_{j-1}$ is removed from agent $i_j$'s bundle where $v_{i_j}(o_{j-1}) = 1$, since $X^{t-1}$ is non-redundant.

For agents not on the transfer path $P$, their bundles remain unchanged. Thus, for each agent $i \in N$, it holds that $v_i(X_i^t) = v_i(X_i^{t-1}) = |X_i^{t-1}| = |X_i^t|$, confirming that $X^t$ is non-redundant.
\end{proof}

    Based on \Cref{app:thm:wefable--iff-no-cycles9} it is established that at the end of every iteration $t\in[T]$, $X^{t}$ is 
    WEF-able. 
    The remaining task is to establish subsidy bounds. 
    
    We focus on two groups: $R$ and $N \setminus R$. 
    
    The selection rule simplifies limit-setting for $R$ and ensures a subsidy bound of 1 (\Cref{app:Alg4_sub_in_game}). 
%
    However, understanding the dynamics of the second group, $N\setminus R$, presents challenges, as the selection rule is not applicable for them.
    For an agent $i \in N \setminus R$, we prove a subsidy bound of $w_i\cdot \frac{1}{w_j}$,
    for some $j\in R$. In particular, the bound is at most $\frac{w_i}{\wmin}$. 

    \begin{observation}
\label{obs:non-redundant}
Let $X$ be any non-redundant allocation.
Let $P=(i,\ldots,j)$ be a path in $G_{X,w}$.
Then $cost_X(P) \leq \frac{|X_j|}{w_j} - \frac{|X_{i}|}{w_{i}}$.
\end{observation}

\label{Omitted Details: binary additive}
\begin{lemma}
\label{lem:non-redundant-difference}
    Let $j\in N$ be any agent, if
    $i\in R(t)$, then 
    $
        \frac{|X^t_j|}{w_j}
        -
        \frac{|X^t_i|}{w_i}
        \leq
        \frac{1}{w_i}.
    $
\end{lemma}
\begin{proof}
If $j$ has never been selected to receive an item, then $|X^t_j|=0$ and the lemma is trivial.

Otherwise, let $t'\leq t$ be the latest iteration in which $j$ was selected.
As agents can not be added to $R$ and by the selection rule, 
$\frac{v_j(X_j^{t'-1}) + 1}{w_j} 
\leq 
\frac{v_i(X_i^{t'-1}) + 1}{w_i}$.
Then by non-redundancy, 
$\frac{|X_j^{t'}|}{w_j} = \frac{v_j(X_j^{t'})}{w_j} = \frac{v_j(X_j^{t'-1}) + 1}{w_j} \leq \frac{v_i(X_i^{t'-1}) + 1}{w_i} = \frac{v_i(X_i^{t'})}{w_i} + \frac{1}{w_i} = \frac{|X_i^{t'}|}{w_i} + \frac{1}{w_i}$. 
As $|X_j^{t}|=|X_j^{t'}|$
and $|X_i^{t}| \geq |X_i^{t'}|$,
the lemma follows.
\end{proof}
From \Cref{lem:non-redundant-difference}, we can conclude the following: 
\begin{proposition}\label{app:Alg4_sub_in_game}
If 
 $i\in R(t)$, 
 then 
 $\ell_i(X^t)\leq \frac{1}{w_i}$.
\end{proposition}
\begin{proof}
Assume $P_i= (i, \ldots, j)$ is the path with the highest total cost starting at $i$ in the $G_{X^t,w}$, i.e., $cost_{X^t}(P_i) = \ell_i(X^t)$. 
\Cref{obs:non-redundant} implies $\ell_i(X^t) \leq \frac{|X_j^t|}{w_{j}} - \frac{|X_i^t|}{w_{i}}$. 
As $i\in R(t)$,
\Cref{lem:non-redundant-difference} implies 
$\frac{|X_j^{t}|}{w_j} - \frac{|X_i^t|}{w_{i}} \leq \frac{1}{w_i}$.
\end{proof}


To prove this upper bound, we need to establish several claims about agents removed from the game. First, we show that an agent removed from the game does not desire any item held by an agent who remains in the game (\Cref{prop: valuation of removed agent}). As a result, these removed agents will not be included in any transfer path (\Cref{prop: removed agent and transfer path}).

Next, we demonstrate that if the cost of a path originating from one of these removed agents at the end of iteration $t$ exceeds the cost at the end of iteration $t' \leq t$ --- the iteration when the agent was removed --- then there exists an edge in this path, $(i_j, i_{j+1})$, such that $v_i(X_j^t) = 0$ (\Cref{at least one 0 valuation}). Based on these claims, we prove that if at the end of iteration $t$, the cost of the maximum-cost path starting from agent removed from the game at $t' < t$ exceeds its cost at $t'$, we can upper-bound it by $\frac{1}{\wmin}$. 

\begin{proposition} \label{prop: valuation of removed agent}
    Let $i$ be an agent removed from the game at the start of iteration $t'$.
Then for all $j \in R(t')$, $v_i(X^{t'}_j)=0$.

Moreover, for all $t > t'$
 and all $j\in R(t)$, 
 $v_i(X^{t}_j)=0$.
\end{proposition}
\begin{proof}
    Suppose that $v_i(X_j^{t'}) \neq 0$. This implies there exists some item $o \in X_j^{t'}$ such that $v_i(o) = 1$. We consider two cases:
    \begin{enumerate}
        \item $o \in X_j^{t'-1}$. In this case, at the start of iteration $t'$, there exists a transfer path from $i$ to $j$.
        Moreover, there is a transfer path from $j$ to $i_0$ at the start of iteration $t'$ (otherwise, $j$ would have been removed from the game at $t'$ as well). 
        Concatenating these paths gives a transfer path from $i$ to $i_0$.
        
        \item $o \not \in X_j^{t'-1}$, that is, $j$ received item $o$ during iteration $t'$, from some other agent $j'$ (where $j' = i_0$ is possible).
        At the start of iteration $t'$, there exists a transfer path from $i$ to $j'$.
        Moreover, there is a transfer path from $j'$ to $i_0$, which is used to transfer the newly allocated item. 
        Concatenating these paths gives a transfer path from $i$ to $i_0$.
    \end{enumerate}
    Both cases contradict the assumption that $i$ was removed at $t'$.

To prove the claim for $t>t'$, we use induction over $t$. 
    
We assume the claim holds for iteration $t-1 > t'$ and prove it for iteration $t$. 
Assume, contrary to the claim, that there exists an agent $i$ who was removed at the start of iteration $t'$, and an agent $j \in R(t)$, such that $v_i(X_j^t) \neq 0$. By the induction hypothesis and the fact that an agent can not be added to $R$, we have $v_i(X_j^{t-1}) = 0$. Therefore, during iteration $t$, $j$ must have received a new item $o_j$ that $i$ values at $1$.

    If $o_j$ was part of $i_0$'s bundle at the start of iteration $t'$, then the transfer path starting at $i$, $(i, i_0)$, must have already existed at the start of iteration $t'$.

    Alternatively, if $o_j$ was originally in another agent's bundle at the start of iteration $t'$, say agent $k \in N$,
    then there must have been an iteration between $t'$ and $t$ in which
    $o_j$ has been transferred from $k$ to another agent, while agent $k$ is compensated by some other item $o_k$, that agent $k$ wants.
    
    If $o_k$ was part of $i_0$'s bundle at the start of iteration $t'$, then the transfer path starting at $i$, $(i, k, i_0)$, must have existed at the start of iteration $t'$. Otherwise, $o_k$ was in another agent's bundle at $t'$, and it, too, would be transferred to a different bundle in a later iteration. 
    
    Since the number of items is finite, this process must eventually lead to an item that was in $X_{i_0}^{t'}$, forming a transfer path starting at $i$ at the start of iteration $t'$ --- a contradiction.
    \end{proof}

\begin{proposition}
    \label{prop: removed agent and transfer path}
    Let $i$ be an agent removed from the game at the start of iteration $t'$. Then, for all $t \geq t'$, $i$ will not be included in any transfer path.
\end{proposition}
\begin{proof}
    First, note that $v_i(X_{i_0}^t) = 0$; otherwise, $i$ would not have been removed at the start of iteration $t'$. Next, any agent $j$, who receives an item from $i_0$'s bundle at iteration $t$, must be in $R(t)$. By \Cref{prop: valuation of removed agent}, $v_i(X_j^t) = 0$, which means that $i$ can not receive any item from $j$, who in $R(t)$, as compensation for another item from their own bundle. Thus, any transfer path in iteration $t$ includes only agents in $R(t)$.
\end{proof}

\begin{proposition} \label{at least one 0 valuation}
    Consider an iteration $t$ and an agent $ i \notin R(t)$, who was removed from the game at the start of iteration $t' < t$.  
Let $ P = (i = i_1, \ldots, i_k) $ be a path in the weighted envy graph starting at $ i $.  
If $cost_{X^t}(P) > cost_{X^{t'-1}}(P) $, then there must exist $ j \in \{1, \ldots, k-1\}$ such that $ i_j \notin R $ at $ t $, $ i_{j+1} \in R(t) $ at $t$, and $ v_{i_j}(X_{i_{j+1}}^t) = 0 $.
\end{proposition}

\begin{proof}
Let $t > t'$ be the earliest iteration in which $cost_{X^t}(P) > cost_{X^{t'-1}}(P)$. This implies that there is at least one agent, say $1 \leq j' \leq k-1$, such that agent $i_{j'+1}$ has received a new item that $i_{j'}$ desires. In other words, $i_{j'+1}$ was part of a transfer path at the start of iteration $t$, and by \Cref{prop: removed agent and transfer path}, $i_{j'+1} \in R(t)$ (in particular, $j' \geq 2$).

Since $i_1 \notin R(t)$ and $i_{j'+1} \in R(t)$, there must exist some $1 \leq p < j'+1$ such that $i_p \notin R(t)$ at $t$ and $i_{p+1} \in R(t)$.
By \Cref{prop: valuation of removed agent}, $v_{i_p}(X^t_{i_{p+1}}) = 0$.
\end{proof}

\begin{proposition} \label{subsidy of agent not in the game}
Let $i\notin R(t)$, be an agent who was removed from the game at the start of iteration $t'<t$.
Then, for $X^t$ the resulting allocation from iteration $t$, $\ell_i(X^{t}) \leq \frac{1}{\wmin}$.
\end{proposition}
\begin{proof}
Denote by $P_i^{t'-1}, P_i^t$ the highest-cost paths starting from $i$ at iterations $t'-1$ (before agent $i$ removed) and $t$, correspondingly. In particular, 
$cost_{X^{t'-1}}(P_i^{t}) \leq cost_{X^{t'-1}}(P_i^{t'-1}) = \ell_i(X^{t'-1})$. 
Moreover, by \Cref{obs:non-redundant} we have $cost_{X^{t'-1}}(P_i^{t'-1}) \leq \frac{|X_{i_{k}}^{t'-1}|}{w_{i_{k}}} - \frac{|X_{i}^{t'-1}|}{w_{i}}$ when $i_k$ is the last agent in $P_i^{t'-1}$. Combined with \Cref{lem:non-redundant-difference}, this gives $\ell_i(X^{t'-1}) \leq \frac{1}{w_i}$.
Therefore, if
$cost_{X^{t}}(P_i^t) = \ell_i(X^t) \leq \ell_i(X^{t'-1})$, then we are done. 

Assume now that $\ell_i(X^t) > \ell_i(X^{t'-1})$. 
Denote the path $P_i^t$ by $(i = i_1,\ldots, i_k)$.

From \Cref{at least one 0 valuation}, there exists $j \in \{1,\ldots,k-1\}$ such that $i_j\notin R(t)$, $i_{j+1}\in R(t)$ and $v_{i_{j}}(X_{i_{j+1}}^t) = 0$.
Then, by \Cref{obs:non-redundant}: 
   \begin{align}
       & \ell_i(X^t) =  cost_{X^t}(i,...,i_j) + cost_{X^t}(i_{j}, i_{j+1}) + cost_{X^t}(i_{j+1},...,i_k) \leq \nonumber \\
       &\leq \left(\frac{|X_{i_j}^t|}{w_{i_j}} - \frac{|X_i^t|}{w_i}\right) + \left(0 - \frac{|X_{i_j}^t|}{w_{i_j}}\right) + \ell_{i_{j+1}}(X^t) \leq  \nonumber\\
       \label{cost of maximum path}
       & \ell_{i_{j+1}}(X^t).
   \end{align}
   Since $i_{j+1} \in R(t)$, it follows from \eqref{cost of maximum path} and \Cref{app:Alg4_sub_in_game} that \[\ell_i(X^t) \leq \ell_{i_{j+1}}(X^t) \leq \frac{1}{w_{i_{j+1}}} \leq 
   \frac{1}{\wmin}.\]
\end{proof}
    
\begin{theorem}\label{app:theorem_21}
    For additive binary valuations, \Cref{alg:binary-additive} computes a WEF-able allocation where the subsidy to each agent $i\in N$ is at most $\frac{w_i}{\wmin}$ in polynomial-time.
    Moreover, the total subsidy is bounded by $\frac{W}{\wmin} - 1$.
\end{theorem}
\begin{proof}
    Together \Cref{app:Alg4_sub_in_game} and \Cref{subsidy of agent not in the game} establish that for every $i\in N$ and $t \in [T]$, $\ell_i(X^t)\leq \frac{1}{\wmin}$. Along with \Cref{app:thm:wefable--iff-no-cycles9}, \Cref{alg:binary-additive} computes a WEF-able allocation $X^{T}$ where the required subsidy per agent $i \in N$ is at most $\frac{w_i}{\wmin}$. 
As there is at least one agent who requires no subsidy (see \Cref{thm:minsubsidy}), the total required subsidy is at most $\frac{W-\wmin}{\wmin} = \frac{W}{\wmin} - 1$.

We complete the proof of \Cref{app:theorem_21} by demonstrating that \Cref{alg:binary-additive} runs in polynomial-time. 

We represent the valuations using a binary matrix $A$ where $v_i(o_j) = 1 \Longleftrightarrow A(i,j) = 1$. Hence, the allocation of items to the bundle of $i_0$ at line 1 can be accomplished in $O(mn)$ time.

    At each iteration $t$ of the while loop, either $X_{i_0}$ or $R$ reduced by 1, ensuring that the loop runs at most $m+n$ times.

    Let $T_v$ represent the complexity of computing the value of a bundle of items, and $T_\phi$ denote the complexity of computing the gain function. Both are polynomial in $m$.

According to \cite{viswanathan2023general}, finding a transfer path starting from agent $i \in N$ (or determining that no such path exists) takes $O(T_v \log m)$. Removing agents at the start of each iteration incurs a complexity of $O(n T_v \log m)$. Furthermore, as stated in \cite{viswanathan2023general}, identifying $u$ requires $O(nT_v)$. Updating the allocation based on the transfer path, according to the same source, takes $O(m)$.

Thus, each iteration has a total complexity of $O(nT_v \log m + nT_v + T_v \log m + m) = O\left( nT_v \log m + m\right)$.

    In conclusion, \Cref{alg:binary-additive} runs in $O\left(\left(m+n\right)\left(nT_v \log m + m\right)\right)$, which is polynomial in both $m$ and $n$.
    \end{proof}
    In Appendix \ref{alg:binary-additive tightness}, we present a tighter bound that is closer to the lower bound, along with a detailed discussion on its tightness.
    
Notice that since the output allocation from \Cref{alg:binary-additive} is non-redundant, $X$ maximizes the social welfare. 
Moreover, as shown in \Cref{app:example: binary additive}, $X$ might not be \emph{WEF(1,0)} (No matter which item is removed from $i_2$'s bundle, $i_1$ still envies). 
However, it is \emph{WEF(0,1)}. 
\begin{proposition} \label{app:binary additive wef01}
    For additive binary valuations, \Cref{alg:binary-additive} computes a WEF$(0,1)$ allocation.
\end{proposition}
\begin{proof}
We prove by induction that at the end of each iteration $t\in [T]$
    $X^t$ satisfies WEF$(0,1)$. This means that for every $i,j \in N$, there exists a set of items $B \subseteq X_j^t$ of size at most $1$ such that $\frac{v_i(X_i^t) + v_i(B)}{w_i} \geq \frac{v_i(X_j^t)}{w_j}$.

The claim is straightforward for the first iteration. We assume the claim holds for the $(t-1)$-th iteration and prove it for the $t$-th iteration. Note that $v_i(X_i^{t-1}) \leq v_i(X_i^{t})$.
\begin{enumerate}
    \item $X_j^t = X_j^{t-1}$ and $X_i^t = X_i^{t-1}$: the claim holds due to the induction step.
    \item $v_i(X_j^t) = v_i(X_j^{t-1})$: This is the case where $j$ was not included in a transfer path, or was included but was not the first agent in the path, and exchanged an item for a new one, both having the same value for $i$. 
    By the induction assumption, there exists some singleton $B^{t-1} \subseteq X_j^{t-1}$
    such that $\frac{v_i(X_i^{t-1}) + v_i(B^{t-1})}{w_i} \geq \frac{v_i(X_j^{t-1})}{w_j} = \frac{v_i(X_j^t)}{w_j}$.
    There exists some singleton $B^{t} \subseteq X_j^{t}$, with $v_i(B^t) = v_i(B^{t-1})$.
    Hence, 
    $\frac{v_i(X_i^{t}) + v_i(B^t)}{w_i} \geq \frac{v_i(X_i^{t-1}) + v_i(B^{t-1})}{w_i} \geq \frac{v_i(X_j^{t-1})}{w_j} = \frac{v_i(X_j^t)}{w_j}$.
    \item  $v_i(X_j^t) < v_i(X_j^{t-1})$: This is the case where $j$ was included in a transfer path but exchanged an item $i$ values for an item $i$ does not value. Then $\frac{v_i(X_i^{t}) + v_i(B)}{w_i} \geq \frac{v_i(X_i^{t-1}) + v_i(B)}{w_i} \geq \frac{v_i(X_j^{t-1})}{w_j} > \frac{v_i(X_j^t)}{w_j}$
        for a set $B \subseteq X_j^t$ of size at most $1$.
    \item $v_i(X_j^t) > v_i(X_j^{t-1})$: there are two subcases:
    \begin{enumerate}
        \item If $j$ is the first agent in the transfer path and received a new item $o$ such that $v_i(o) = 1$, then 
        $\frac{v_i(X_i^{t-1}) + 1}{w_i} \geq \frac{v_j(X_j^{t-1}) + 1}{w_j}$ due to the selection rule,
        and 
        $\frac{v_j(X_j^{t-1}) + 1}{w_j}
        \geq \frac{v_i(X_j^{t-1}) + 1}{w_j}$ due to non-redundancy.
        
        We can conclude that $\frac{v_i(X_i^{t}) + 1}{w_i} \geq \frac{v_i(X_i^{t-1}) + 1}{w_i} \geq \frac{v_i(X_j^{t-1}) + 1}{w_j} = \frac{v_i(X_j^{t})}{w_j}$. The claim holds for $B = \{o\}\subseteq X_j^t$.
        \item If $j$ was not the first agent in the path, but exchanged an item that $i$ does not value for an item $o$ that $i$ values, $v_i(o) = 1$.
        Let $t' < t$ represent the most recent iteration in which agent $j$ was selected and received a new item. Note that $\frac{v_i(X_i^{t}) + 1}{w_i} \geq \frac{v_i(X_i^{t'-1}) + 1}{w_i} \geq \frac{v_j(X_j^{t'-1}) + 1}{w_j}$ due to the selection rule.
        
        Assume to the contrary that $\frac{v_i(X_i^t) + 1}{w_i} < \frac{v_i(X_j^t)}{w_j}$. Then,
        \begin{align*}
            & \frac{v_i(X_i^t) + 1}{w_i} < \frac{v_i(X_j^t)}{w_j} \leq \frac{v_j(X_j^{t})}{w_j} = \frac{v_j(X_j^{t'})}{w_j} = \\ &\frac{v_j(X_j^{t'-1}) + 1}{w_j} \leq \frac{v_i(X_i^{t'-1}) + 1}{w_i} \leq \frac{v_i(X_i^{t}) + 1}{w_i},
        \end{align*}
        a contradiction.
        Hence, $\frac{v_i(X_i^t) + 1}{w_i} \geq \frac{v_i(X_j^t)}{w_j}$ and the claim holds for $B = \{o\}\subseteq X_j^t$.
    \end{enumerate}
\end{enumerate}
\end{proof}

\subsection{Matroidal Valuation}
\label{sec:matroidal}
 Another valuation class,  slightly more general than binary additive, 
is the class of \emph{matroidal valuations}, also called \emph{matroid rank valuations}~ (\cite{BarmanVermaAAMAS2021,benabbou2020}). In this section we prove that the subsidy bound provided by \Cref{alg:binary-additive} for agents with binary additive valuations, which is $\frac{W}{\wmin} - 1$ (\Cref{app:theorem_21}),
is not applicable for agents with matroidal valuations. Specifically, we show a lower bound of $\frac{m}{n}\left(\frac{W}{\wmin} - n\right)$, which is linearly increasing with $m$.

A matroidal valuation is based on a matroid ${\mathcal F}$ over $M$.%
\footnote{
Formally, ${\mathcal F}$ is a family of subsets of $M$ such that
$\emptyset \in {\mathcal F}$, $F' \subseteq F \in {\mathcal F}$ implies $F' \in {\mathcal F}$, and for any $F, F' \in {\mathcal F}$ where 
$|F|>|F'|$, there exists an item $o \in F\setminus F'$ s.t. $F'\cup\{o\} \in {\mathcal F}$.
}
Then, the value of each subset $A\subseteq M$ equals $\max_{F \in {\mathcal F}} |A\cap F|$. 

Note that a matroidal valuation is submodular, but not necessarily additive.
A binary additive valuation is a special case of a matroidal valuation, in which each agent $i$ has ${\mathcal F}_i= \{F \mid F\subseteq B_i\}$, where $B_i:=\{o \mid o \in M, v_i(o)=1\}$.

We demonstrate that the total subsidy bound derived by \Cref{alg:binary-additive} for agents with binary additive valuations might not hold for agents with matroidal valuations.

\begin{theorem}\label{thm:subsidy-lb-matroidal}
There exists an instance with matroidal valuations for which, in any WEF non-wasteful allocation, 
the subsidy for some agent $i\in N$ is at least $\frac{m}{n} 
(\frac{w_{i}}{\wmin}-1)$, and the total subsidy is at least $\frac{m}{n}\left(\frac{W}{\wmin} - n\right)$.
\end{theorem}
\begin{proof}
Assume there are $n$ agents with weights $w_{1} \leq \ldots \leq w_n$. 
There are $m=nk$ items, and the valuation of each agent $i$ for bundle $X_i$ 
is given by $\min(k, |X_i|)$, i,e., each agent values at most $k$ items
(We can assume that the feasible bundles are defined as ${\mathcal F} = \{F \mid F\subseteq M, |F|\leq k\}$). 

In this setting, the only non-wasteful allocation is one in which each agent is allocated exactly $k$ items.
An agent $i\in N$ envies another agent $j\neq i \in N$ by amount given by 
$\frac{k}{w_j} - \frac{k}{w_i}$. Equivalently, the cost of any path $P=(i,\ldots, j)$ for some $j \neq i \in N$ is
$
    k\left(\frac{1}{w_{j}} - \frac{1}{w_i}\right).
$
The maximum-cost path starting from agent $i$ ends at agent $1$, and its cost is
$
     k\left(\frac{1}{\wmin} - \frac{1}{w_i}\right).
$
To eliminate envy, each agent $i\in N$ must receive a subsidy of 
$
k\left(\frac{w_i}{\wmin} - 1\right) = \frac{m}{n}\left(\frac{w_i}{\wmin} - 1\right).
$
The total subsidy required is 
\begin{align*}
   &\sum_{1\leq i \leq n}\frac{m}{n}\left(\frac{w_i}{\wmin} - 1\right) =
   \sum_{1 < i \leq n}\frac{m}{n}\left(\frac{w_i}{\wmin} - 1\right) =  \frac{m}{n}\left( \frac{W - \wmin}{\wmin} - \left(n - 1\right) \right) = \\ &\frac{m}{n}\left( \frac{W}{\wmin} - n \right)
   .
\end{align*}
\end{proof}
Importantly, whereas the upper bound for binary additive valuations $\frac{W}{\wmin} - 1$ is independent of $m$, the lower bound for matroidal valuations is increasing with $m$.

\section{Additive Valuations and Identical Items} 
\label{sec:identical_items}
With a slight abuse of notation, we denote by $v_i$ the value of agent $i$ to a single item.
Thus, for an allocation $X$, if $|X_i| = m_i$, then $v_i(X_i) = v_i m_i$. 
We denote an allocation by a tuple of integers  $(m_1,m_2,\ldots, m_n)$,  representing the numbers of items allocated to the agents. Note that $m=m_1+\cdots+m_n$.


In this section, for simplicity, we order the agents in descending order of their \emph{value} rather than their weight, that is,
we sort the agents so that $(V \ge)~v_1 \ge v_2 \ge \ldots \ge v_n$.

We show matching upper and lower bounds on the subsidy per agent and on the total subsidy.

The following lemma states that weighted envy-freeability can be characterized by the weighted reassignment-stability condition for swapping only a pair of two agents.
\begin{lemma}
\label{lem:WEFability-additive-identical-items}
For additive valuations with identical items, 
an allocation $X = (m_1,m_2\dots, m_n)$ is WEF-able if and only if 
for each $1 \le i,j \le n$ with $v_i < v_j $ we have 
$\frac{m_i}{w_i} \le \frac{m_j}{w_j}$. 
\end{lemma}

\begin{proof}
We first show the only-if part.
Assume that $(m_1,\ldots, m_n)$ is WEF-able.
Then, by Theorem~\ref{th:ef}, it is weighted reassignment-stable.
For the permutation that only swaps $i$ and $j$, inequality~\eqref{eq:WRS} in the definition of weighted reassignment-stability implies that 
$\frac{v_i \cdot m_i}{w_i} + \frac{v_j \cdot m_j}{w_j} \ge \frac{v_i \cdot m_j}{w_j} + \frac{v_j \cdot m_i}{w_i}$.
When $v_i < v_j$, this implies that $\frac{m_i}{w_i} \le \frac{m_j}{w_j}$.

We now show the if part.
WLOG, we can assume when $i> j$ and $v_i=v_j$, 
$m_i/w_i \ge m_j/w_j$ holds (otherwise we can rename agents' identifiers). 
We will show that for any $i > j$ 
we have 
\begin{align}
\label{eq:cost_pij}
    &cost_X(i,j)\le cost_X(i, i-1)+cost_X(i-1,i-2)+\dots +cost_X(j+1,j)
\end{align}

Once this is done, we can show that any cycle in the weighted envy-graph has a non-positive cost as follows.
Let $C$ be a cycle in the weighted envy-graph.
We partition the set of edges of $C$ into the two sets of ``ascending'' edges and ``descending'' edges: 
\begin{align*}
E(C) = E_+(C) \cup E_-(C),
\end{align*}
 where 
 \begin{align*}
 E_+(C) =\{ (i,j) \in E(C) \mid i>j \} \text{ and }
E_-(C) =\{ (j,i) \in E(C) \mid i>j \}. 
 \end{align*}

To show $C$ has non-positive 
cost is to show 
\begin{align*}
\sum_{(i,j)\in E_+(C)}cost_X(i,j)+\sum_{(j,i)\in E_-(C)}cost_X(j,i) \le 0    
\end{align*}
The inequality \eqref{eq:cost_pij} implies
\begin{align*}
 &\sum_{(i,j)\in E_+(C)}cost_X(i,j)
 \leq \sum_{(i,j)\in E_+(C)}(cost_X(i,i-1)+\dots +cost_X(j+1,j)).
\end{align*}
Let $E'_+(C)$ be a \emph{multiset} of edges of $C$ defined as 
\begin{align*}
E'_+(C)=\cup_{(i,j)\in E_+(C)}\{ (i,i-1),(i-1,i-2),\dots ,(j+1,j) \}.    
\end{align*}
Then it suffices to show that 
\begin{align*}
\sum_{(i,j)\in E'_+(C)}cost_X(i,j)+\sum_{(j,i)\in E_-(C)}cost_X(j,i) \le 0.    
\end{align*}
Now, to each edge $(j_0,i_0) \in E_-(C)$,
we assign the subset of edges 
$\{ (i_0,i_0-1),\dots, (j_0+1,j_0) \}$ from $E'_+(C)$;
this assignment partitions $E'_+(C)$ into disjoint subsets,  since $C$ is a cycle.

Then it suffices to show that the sum of costs of ``ascending''  edges assigned to $(j_0,i_0) \in E_-(C)$ plus the cost of $(j_0,i_0)$ itself (i.e., $cost_X(i_0,i_0-1)+cost_X(i_0-1,i_0-2)+\dots +cost_X(j_0+1,j_0)+cost_X(j_0,i_0)$) is non-positive, since $\sum_{(i,j)\in E'_+(C)}cost_X(i,j)+\sum_{(j,i)\in E_-(C)}cost_X(j,i)$ is the sum of these values.

Indeed, we have
\begin{align*}
\tiny
&cost_X(i,i-1)+ \dots 
 +cost_X(j+1,j)
+cost_X(j,i)\\
&=v_{i}\left(\frac{m_{i-1}}{w_{i-1}}-\frac{m_{i}}{w_{i}}\right)+\dots 
+v_{j+1}\left(\frac{m_{j}}{w_{j}}-\frac{m_{j+1}}{w_{j+1}}\right)   + v_{j}\left(\frac{m_{i}}{w_{i}}-\frac{m_{j}}{w_{j}}\right)
\\
&
\le v_{j}\left(\frac{m_{i-1}}{w_{i-1}}-\frac{m_{i}}{w_{i}}\right)+\dots 
+v_{j}\left(\frac{m_{j}}{w_{j}}-\frac{m_{j+1}}{w_{j+1}}\right) + v_{j}\left(\frac{m_{i}}{w_{i}}-\frac{m_{j}}{w_{j}}\right)\\
&=v_{j}\left(\frac{m_{j}}{w_{j}}-\frac{m_{i}}{w_{i}}\right) + v_{j}\left(\frac{m_{i}}{w_{i}}-\frac{m_{j}}{w_{j}}\right)= 0,
\end{align*}
where we use the fact that $v_i \ldots v_{j+1} \le v_{j}$ and $\frac{m_{k}}{w_{k}} \ge \frac{m_{k+1}}{w_{k+1}}$ for $i \le k \le j-1$ in the inequality.
Therefore, any cycle in the weighted envy-graph has a non-positive cost, and the allocation is WEF-able by Theorem~\ref{th:ef}.

It remains to prove
the inequality \eqref{eq:cost_pij}. 
We show this by induction on $j-i$.
If $i-j = 1$, then $cost_X(i,j) = cost_X(i,i-1)$ holds trivially.

Assume $i-j > 1$.
By the inductive hypothesis, we have 
\begin{align*}
&
cost_X(i,j+1)\le cost_X(i,i-1)
+cost_X(i-1,i-2)
+\dots +cost_X(j+2,j+1).    
\end{align*}
Hence, it suffices to show that $cost_X(i,j)\le cost_X(i,j+1)+cost_X(j+1,j)$.
Indeed, 
\begin{align*}
cost_X(i,j) &= v_{i}\left(\frac{m_{j}}{w_{j}}-\frac{m_{i}}{w_{i}}\right)= v_{i}\left(\frac{m_{j}}{w_{j}}-\frac{m_{j+1}}{w_{j+1}}\right) + v_{i}\left(\frac{m_{j+1}}{w_{j+1}}-\frac{m_{i}}{w_{i}}\right)\\
&\le v_{j+1}\left(\frac{m_{j}}{w_{j}}-\frac{m_{j+1}}{w_{j+1}}\right) + v_{i}\left(\frac{m_{j+1}}{w_{j+1}}-\frac{m_{i}}{w_{i}}\right)= cost_X(j+1,j)+cost_X(i,j+1),
\end{align*}
where we use the fact that $v_i \le v_{j+1}$ and $\frac{m_{j+1}}{w_{j+1}} \le \frac{m_{j}}{w_{j}}$ in the inequality.
\end{proof}

Using this \Cref{lem:WEFability-additive-identical-items}, we first prove a lower bound.

\begin{theorem}\label{thm:subsidy-lb-additive-identical-items}
Even with additive valuations and identical items, it is impossible to guarantee that the subsidy per agent is smaller than $V w_i \sum_{1 \le j < i} \frac{1}{w_j}$, or that the total subsidy is smaller than $V \sum_{2 \leq i \leq n} \left( w_i \sum_{1 \le j < i} \frac{1}{w_j} \right)$.
\end{theorem}


\begin{remark}
\Cref{worst case allocation can be chosen with weights}
shows a lower bound for general additive valuations, which is worse than the lower bound of \Cref{thm:subsidy-lb-additive-identical-items} for identical additive valuations.
This is because the bound of \Cref{worst case allocation can be chosen with weights} holds for any weight vector, whereas the bound of \Cref{thm:subsidy-lb-additive-identical-items} holds only for weight vectors where weights are not integer multiples of adjacent weights (see the proof).
\end{remark}

\begin{proof}
We construct an instance in which $v_i = V + (n-i)\cdot \epsilon$ for some small $\epsilon > 0$. Note that $v_1 > \cdots > v_n$.
The weights are in the same order, that is, $\wmax = w_n \ge \cdots \ge w_1 = \wmin \ge 2$. 

The smallest weight $w_n$ is an integer. For any $i \in \{2, \ldots, n\}$, $w_i = \left(k_i - \epsilon \right) w_{i-1}$, for some integer $k_i \geq 2$ and some small $\epsilon>0$.

We aim to determine the number of items $m = \sum_{i \in N} m_i$ such that agent $n$ receives exactly $w_n$ items, and each agent $i \in \{2, \ldots, n\}$ envies agent $i-1$. 

By the WEF property, each agent $i$ must achieve a total utility of approximately $w_i v_i \sim w_i V$. Therefore, the total subsidy required is given by:
\[
\left(W - m\right)V = \sum_{i \in N} \left(w_i - m_i\right)V,
\]
The term $m V$ accounts for the utility generated by allocating $m$ items to the agents. 

Our goal is to find values $m_1, \ldots, m_n$ such that agent $i$ receives a subsidy of $\left(w_i - m_i\right)V$. 

First, note that for each $i \in \{2, \ldots, n\}$, we have $k_i = \frac{w_i}{w_{i-1}} + \epsilon \sim \frac{w_i}{w_{i-1}}$, where $\epsilon$ is small.

We proceed by allocating items to agents as follows:
\begin{enumerate}
    \item Agent $1$ receives $m_1 = w_1$ items, as assumed, and receives no subsidy.
    \item Agent $2$ receives 
    $m_{2} = k_{2} \left( m_1 - 1 \right) \sim \frac{w_{2}}{w_1} \left( m_1 - 1 \right) = w_{2} - \frac{w_2}{w_1}$ items. The subsidy for agent $2$ is $w_{2} v_{2} \left( \frac{m_1}{w_1} - \frac{m_{2}}{w_{2}} \right) \sim V w_{2} \frac{1}{w_1}$.
    \item Agent $3$ receives 
    $m_{3} = k_{3} \left( m_{2} - 1 \right) \sim \frac{w_{3}}{w_{2}} \left( m_{2} - 1 \right) = w_{3} - \frac{w_{3}}{w_n} - \frac{w_{3}}{w_{2}}.$ The subsidy for agent $3$ is $w_{3} \left( v_{2} \left( \frac{m_{1}}{w_{1}} - \frac{m_{2}}{w_{2}} \right) + v_{3} \left( \frac{m_{2}}{w_{2}} - \frac{m_{3}}{w_{3}} \right) \right) \sim V w_{3}\left(\frac{1}{w_{2}} + \frac{1}{w_1}\right)$.
    \item In general, for any agent $i \in \{2, \ldots, n\}$, the number of items they receive is:
    $m_i = k_i \left( m_{i-1} - 1 \right) \sim \frac{w_i}{w_{i-1}} \left( m_{i-1} - 1 \right) = w_i - \left( w_i \sum_{1 \le j < i} \frac{1}{w_j} \right).$ The subsidy for agent $i$ is $\sim V w_{i}\sum_{1 \le j < i} \frac{1}{w_j}$.
\end{enumerate}

By \Cref{lem:WEFability-additive-identical-items}, we can not remove any item from any agent $i \in \{1, \ldots, n-1\}$ because:
\[
\frac{m_{i+1}}{w_{i+1}} = \frac{k_{i+1} \left( m_i - 1 \right)}{w_{i+1}} > \frac{w_{i+1}}{w_i} \frac{m_i - 1}{w_{i-1}} = \frac{m_i - 1}{w_i}.
\]
This inequality ensures that reducing the number of items allocated to any agent would violate the WEF condition.

Thus, this allocation satisfies WEF with the smallest possible subsidy, as required, and the total required subsidy is 
\[
V \sum_{2\leq i \leq n} \left( w_i \sum_{1 \le j < i} \frac{1}{w_j} \right)
\]
\end{proof}


We now establish a matching upper bound on the subsidy. Our algorithm is defined in \Cref{alg:n-subsidy-additive-identical-items}.
First, the agents are sorted by their valuations such that $v_1 \ge v_2 \ge \ldots \ge v_n$. The number of items of each agent $i$ is initialized to $m_i=0$.
Then, while there are unallocated items, the algorithm finds the agent $i \in \{2, \ldots, n\}$ with the maximum index such that $\frac{1+m_i}{w_i} \le \frac{m_{i-1}}{w_{i-1}}$. If no such agent exists, the algorithm, selects agent $1$. The chosen agent $i\in N$ receives a new item, i.e., $m_i$ increases by 1.

\begin{algorithm}[t]
\caption{
Weighted Sequence Protocol For Additive Valuations and identical items}
\label{alg:n-subsidy-additive-identical-items}
\KwIn{ Instance $(N,M,v, \mathbf{w})$ with additive valuations and identical items.}
\KwOut{ WEF-able allocation $(m_1,\ldots,m_n)$ with required subsidy of at most $V w_{i}\sum_{1 \le j \le i}\frac{1}{w_{j}}$ per agent,
and total subsidy at most
 $V \sum_{2\leq i \leq n} \left( w_i \sum_{1\leq j \leq i} \frac{1}{w_j} \right)$.}
$m_i \gets 0$ for each $i\in N$\;
Sort the agents such that $v_1\geq \cdots \geq v_n$.\;
\For{$o:$ 1 to $m$}{
Let $N'\gets \{i\in \{2, \ldots, n\} \mid \frac{1+m_i}{w_i}\le \frac{m_{i-1}}{w_{i-1}}\} \cup \{1\}$  
\tcc{We always have $1\in N'$}
Let $u\gets\max_{i \in N'} i$\;
$m_u \leftarrow m_u + 1$\;
}
\Return $(m_1, \ldots, m_n)$
\end{algorithm}

\begin{theorem}
\label{thm:subsidy-ub-additive-identical-items}
\Cref{alg:n-subsidy-additive-identical-items} 
outputs a WEF allocation with subsidy at most $w_i V \sum_{1\leq j \leq i} \frac{1}{w_j}$ 
for each agent,
and total subsidy at most $V \sum_{2\leq i \leq n} \left( w_i \sum_{1\leq j \leq i} \frac{1}{w_j} \right)$.
\end{theorem}
\begin{proof}
We first prove that the allocation output by \Cref{alg:n-subsidy-additive-identical-items} is WEF-able by \Cref{lem:WEFability-additive-identical-items}.
The definition of $N'$ ensures that, if agent $u\geq 2$ receives an item, then after the update, $\frac{m_u}{w_u} \leq \frac{m_{u-1}}{w_{u-1}}$.
Therefore, this condition holds throughout the algorithm for all adjacent pairs of agents.
By transitivity, 
$\frac{m_u}{w_u} \leq \frac{m_{j}}{w_{j}}$ for all $u>j\geq 1$.

We prove that while running the algorithm, 
the cost of the highest-cost path is always bounded by $V \sum_{1\leq j \leq i} \frac{1}{w_j}$. 

From the proof of Lemma~\ref{lem:WEFability-additive-identical-items}, WLOG, we can assume the highest-cost path from each agent $i$
is $P_i = (i, i-1, \ldots, 1)$, since $cost_X(j, j-1)$ is non-negative, $cost_X(j, k)$ where $j<k$ is non-positive, and there exists no positive cycle. 
Also, $cost_X(i, j)$ where $j\geq i+2$ is smaller than or equal to $\sum_{i < k \le j}cost_X(k, k-1)$.
In other words, $\ell_i(X) = \sum_{1 < j \le i} cost_X(j,j-1)$.

Now, we prove for each agent $i\in N$, $\ell_i(X) \le V \sum_{1\leq j \leq i} \frac{1}{w_j}$.
The proof is by induction over the iteration.

When there is no allocated item, this must be true. 
Now, suppose that the algorithm allocates an item to agent $u$.
We consider three cases.

\underline{Case 1:} $u > i$.
The cost of the highest-cost path starting at $i$ is not affected by allocating the item to $u$ in this iteration, as $\ell_i(X) = \sum_{1 < j \le i} cost_X(j,j-1)$. By the induction assumption, $\ell_i(X) \leq V \sum_{1 < j \le i} \frac{1}{w_j}$. 

\underline{Case 2:}  $1 > i\ge u$.
Then $cost_X(u+1, u)$ increases by at most $\frac{v_{u+1}}{w_u}$ (or does not increase at all if $u=i$), while $cost_X(u, u-1)$ decreases by $v\frac{v_{u}}{w_u}$. 
The costs of all other edges remain unchanged.
Thus, the total path cost weakly decreases.
By the induction assumption, $\ell_i(X) \leq V \sum_{1 < j \le i} \frac{1}{w_j}$. 

\underline{Case 3:} $u=1$. 
Then, by the fact that the algorithm chose $1$, before the allocation
we must have  $\frac{1+m_i}{w_i}>\frac{m_{i-1}}{w_{i-1}}$ for each $i>1$. Otherwise, the algorithm would have chosen some agent $i > 1$ to allocate the item. 

Thus, for each $i>1$, before the allocation,
\begin{align*}
&cost_X(i,i-1) =
\frac{v_i m_{i-1}}{w_{i-1}} - \frac{v_i m_i}{w_i} < v_i\left(\frac{1+m_i}{w_i}  - \frac{m_i}{w_i}\right)
= \frac{v_i}{w_i} \leq \frac{V}{w_i}.
\end{align*}
This implies that before allocating the item to agent $1$: $\ell_i(X) = \sum_{1< j \le i} cost_X(j,j-1) < \sum_{1< j \le i} \frac{1}{w_j}.$
After allocating the item to agent $1$, the path cost increases by at most $\frac{V}{w_{1}}$ 
Therefore, $\ell_i(X) < V\sum_{1 \le j \le i} \frac{1}{w_j}.$

Now we can conclude that the cost of a highest-cost path is bounded by $V \sum_{1 \le j \le i}\frac{1}{w_{j}}$. 
Thus, the subsidy for agent $i$ is at most $w_i\left(\sum_{1 \le j \le i}\frac{1}{w_{j}}\right).$

Note that agent $2$ receives a subsidy of $p_{2} = v_{2} w_{2} \left( \frac{m_{1}}{w_{1}} - \frac{m_{2}}{w_{2}} \right),$ which ensures that \begin{align*}
    &\frac{v_{2} m_{2} + p_{2}}{w_{2}} = \frac{v_{2} m_{2}}{w_{2}} + \frac{v_{2} m_{1}}{w_{1}} - \frac{v_{2} m_{2}}{w_{2}} = \frac{v_{2} m_{1}}{w_{1}}.
\end{align*}
Thus, agent $1$ must receive no subsidy, because otherwise, the subsidy $p_{2}$ would not eliminate the envy of agent $2$ towards agent $1$.

Therefore, 
the total subsidy is at most $V \sum_{2\leq i \leq n} \left( w_i \sum_{1\leq j \leq i} \frac{1}{w_j} \right).$
\end{proof}

\cite{Brustle2020} proved that when the weights are equal, the total subsidy required for agents with additive valuations is bounded by $(n-1)V$. The following proposition demonstrates that with different weights, even when the weights are nearly equal, the total subsidy bound can be linear in $n^2$.

\begin{proposition} \label{subsidy bound for identical items and almost equal weights}
    Even when the weights are nearly equal, the total subsidy lower bound for agents with additive valuations and identical items is linear in $n^2$.
\end{proposition}
\begin{proof}
    Assume that $w_1 = 2n$ and $w_i = w_{i+1} - \epsilon$ for each $i \in {2, \ldots, n-1}$, with $\epsilon > 0$. Note that all weights are positive, nearly equal, and $w_1 > w_2 > \ldots > w_n$.

Let $m_1 = w_1$ and 
$m_i = m_{i-1} - 1$ for $i\in \{2, \ldots, n\}$.
Let $m := \sum_i m_i = n\cdot(3n+1)/2$.

    By \Cref{lem:WEFability-additive-identical-items}, 
agent $i \in {2, \ldots, n}$ can receive at most $m_i = m_{i-1} - 1$ items. Thus, the subsidy for agent $i\in \{2, \ldots, n\}$ is at least $w_i V \sum_{1\le j < i } \frac{1}{w_j} \sim (i-1)V$.
    Hence, the total subsidy is at least $\sim V\sum_{2 \le i \le n} i-1 = V\sum_{1 \le i \le n-1} i = \frac{(n-1)(n-2)}{2}V = \Omega(n^2 V)$.
\end{proof}

\subsection{Optimal Algorithm for a Special Case of Additive Valuations and Identical Items}
In most settings studied in this paper, computing the optimal subsidy is NP-hard.
However, in the special case of identical items and additive valuations, we have a polynomial-time algorithm.

Here, for convenience, we assume that $v_1 < v_{2} < \dots < v_{n}$, contrary to the assumption in the previous subsection.
We   propose an algorithm based on dynamic programming
that computes an allocation with the minimum total subsidy.

We say that an allocation $(m_1,\dots, m_n)$ is \emph{feasible} if $m_1 + \dots + m_n = m$.

For $1 \le i \le n$, $0 \le j \le m$, and $0 \le m_i \le m$, 
let $T(i,j,m_i)$ be defined as the minimum total subsidy when 
\begin{itemize}
	\item the agents are restricted to $1, 2,\dots, i$,
	\item the number of items is $j$, and 
	\item the number of items allocated to agent $i$ is $m_i$.
\end{itemize}
Then the minimum total subsidy we want to compute equals $\min_{1 \le m_i \le m}T(n,m,m_i)$.

\begin{lemma}\label{lem:identical-item-correctness}
The following recursive formula holds.
\begin{align}\label{eq:identical-item-recursive-formula}
\resizebox{\textwidth}{!}{
$T(i,j,m_i) = 
\begin{cases}
0 & \text{$i=1$ and $j=m_i$}\\
\underset{0 \le m_{i-1} \le \min\left(j-m_i,\frac{w_{i-1}}{w_i}m_i\right)}{\min} \left\{ T(i-1,j-m_i,m_{i-1}) + {\displaystyle \left(\sum_{i'=1}^{i-1}w_{i'}\right)\cdot \left(\frac{m_i}{w_i}-\frac{m_{i-1}}{w_{i-1}}\right)v_{i-1}} \right\} & \text{$i\ge 2$ and $\frac{w_i}{\sum_{i'=1}^{i}w_{i'}}j \le m_i \le j$}\\
\infty & \text{otherwise}.
\end{cases}
$
}
\end{align}
\end{lemma}

\begin{proof}
When $i=1$, i.e., the number of agents in the market is one, an allocation is feasible iff $j=k$, and the total subsidy is zero for the feasible allocation.

Assume that  $i \ge 2$.
when agent $i$ is added to the market with $i-1$ agents, if $m_i$ (resp., $m_{i-1}$) items are allocated to $i$ (resp., $i-1$), then the cost of edge $(i-1,i)$ in the weighted-envy graph is $\left(\frac{m_i}{w_i}-\frac{m_{i-1}}{w_{i-1}}\right)v_{i-1}$. Since a highest-cost path from agent $i'(<i)$ is $(i', i'+1 , \dots, n)$ from the proof of \Cref{thm:subsidy-ub-additive-identical-items}, $\sum_{i'=1}^{i-1}w_{i'}$ times the cost of edge $(i-1,i)$ is added to the total subsidy.
Moreover, since the cost of the edges before $i$ does not affect any highest-cost path path starting from $i$, allocations for agents $i' \le i-2$ achieving $T(i-1,j-m_i,m_{i-1})$ do not affect the subsidies for agents $i'' \ge i-1$.
Hence, one can optimize the subsidies for agents $i' \le i-2$ in $T(i-1,j-k,k')$ and Equation~\ref{eq:identical-item-recursive-formula} is correct.

We note that 
by \Cref{lem:WEFability-additive-identical-items},
$\frac{m_i}{w_i} \geq 
\frac{\sum_{i'=1}^i m_{i'}}{\sum_{i'=1}^i w_{i'}} 
$.
Hence, 
if $m_i < \frac{w_i}{\sum_{i'=1}^{i}w_{i'}}j$, then for an allocation to be WEF-able, $\sum_{i'=1}^{i}m_{i'}$ should be less than $j$ and thus there exists no WEF-able and feasible allocation. 
Therefore, we require $m_i \ge \frac{w_i}{\sum_{i'=1}^{i}w_{i'}}j$ in the second case of Equation~\ref{eq:identical-item-recursive-formula}.
Moreover, for an allocation to be WEF-able, $m_{i-1}$ should be less than or equal to $\frac{w_{i-1}}{w_i}m_i$ and thus we require $m_{i-1} \le \frac{w_{i-1}}{w_i}m_i$ in the second case of Equation~\ref{eq:identical-item-recursive-formula}.
\end{proof}

\begin{theorem}\label{thm:identical-item-optimal-algorithm}
There exists a polynomial time algorithm to compute an allocation with the minimum total subsidy for additive valuations and identical items if the valuations of each agent are all different.
\end{theorem}
\begin{proof}
The minimum subsidy equals $\min_{1 \le k \le m}T(n,m,k)$ from Lemma~\ref{lem:identical-item-correctness}. 
We can also compute an allocation achieving the minimum by keeping the value $m_{i-1}$ attaining the minimum in Equation~\ref{eq:identical-item-recursive-formula}.
The size of the table $T$ is ${\rm O}(nm^2)$, and it takes ${\rm O}(m)$ time to fill each cell of the table.
Therefore, the total running time is ${\rm O}(nm^3)$.
\end{proof}

\section{Monetary Weighted Envy-Freeness}
\label{sec:MWEF}
In practice, the available subsidy may be smaller than what's required for WEF. A natural question is what relaxation of fairness can be achieved. One solution is to allocate the subsidy only to agents whom nobody envies, preventing additional envy. This motivates the following concept:
\begin{definition}
	An outcome $(X,\mathbf{p})$ is called \emph{monetarily weighted envy-free (MWEF)} if 
    $p_j = 0$ for all $j\in N$ such that some agent $i\in N$ has weighted envy towards $j$.
\end{definition}
 Note that the definition avoids the issue of whether the indivisible item allocation is fair.
MWEF formalizes the idea that a limited subsidy is used effectively to improve fairness, regardless of whether the allocation 
$X$ is ``good'' or ``bad''.
Denote the total amount of money available for subsidy by $d$.
If $d=0$, so there is no subsidy at all ($p_i=0$ for all $i\in N$), then the allocation is vacuously MWEF. 
Also, every WEF allocation (with or without subsidy) is MWEF.
  
Our main result in this section is that MWEF can be achieved using any WEF-able allocation and any total subsidy amount $d$.
As far as we know, this algorithm is new even for the unweighted setting.


\begin{theorem}
There is a polynomial-time algorithm that, 
for any instance with monotone valuations, 
given any WEF-able allocation $X$ and any amount of money $d$, finds a subsidy vector $\mathbf{p}$ with $\sum_i p_i = d$, such that $(X,\mathbf{p})$ is MWEF.
\end{theorem}
\begin{proof}
For any given outcome $(X,\mathbf{p})$, the corresponding \emph{weighted envy-graph respecting the subsidy}, denoted $G_{X,w,\mathbf{p}}$, is a complete directed graph with vertex set $N$. For any pair of agents $i,j\in N$, 
$cost_X(i,j) \ =\  \frac{v_i(X_j)+p_j}{w_j}-\frac{v_i(X_i)+p_i}{w_i}$, represents the envy agent $i$ has for agent $j$ under $(X,\mathbf{p})$.

Let $\ell_i(X)$ be the maximum cost of any path in the weighted envy-graph $G_{X,w}$ that starts from $i$. 

If the total money $d$ is at exactly $\sum_{i\in N}\ell_i (X)\cdot w_i$, we can let each agent $i$'s subsidy be $p_i=\ell_i(X) \cdot w_i$. The outcome is WEF by \Cref{thm:minsubsidy},  hence also MWEF.

If $d>\sum_{i\in N}\ell_i(X)\cdot w_i$, we first allocate $p_i=\ell_i(X)\cdot w_i$, and then allocate the surplus amount in proportion to the weights of the agents, which is WEF as well. 

The challenging case is if $d<\sum_{i\in N}\ell_i(X) \cdot w_i$. In this case, we initialize $p_i=0$ for all $i\in N$, and then gradually increase the subsidies of some agents as follows.

Let $\ell_i(X,\mathbf{p})$ be the maximum cost of any path in the weighted envy-graph $G_{X,w,\mathbf{p}}$ that starts at $i$. 
We identify the set of agents $N^*$ who have the highest $\ell_i(X,\mathbf{p})$. 

Note that there is no agent $j$ outside $N^*$ who has zero or positive edge to any agent in $N^*$ because if this is the case, $j$ would be in $N^*$. Therefore, any agent $i\in N^*$ can be given a tiny amount of money without leading some agent outside $N^*$ to become envious. 
Moreover, no agent $i\in N^*$ has a strictly positive edge to $k\in N^*$ or else $k$ would not be a part of $N^*$. 

When we allocate the money in proportion to the weights, all the $\ell_i(X,\mathbf{p})$ for $i\in N^*$ decrease at the same rate (so they remain maximum), and $\ell_j(X,\mathbf{p})$ for $j\not \in N^*$ might increase.

As we do this, the set $N^*$ may increase. Eventually, all the money is allocated.  
\end{proof}


\section{Experiments}
\label{sec:exp}
In this section, we compare the minimum subsidy required for weighted envy-free allocation and the subsidy obtained by our proposed algorithms, along with their theoretical guarantees. We generate synthetic data for our experiments. We consider $n\in\{5,8,10\}$ agents and choose the number of items $m \in \{n, 2n, 3n, 4n, 5n\}$ and fix the weight vector $\mathbf{w}=(1,2, \ldots ,n)$.  For each agent-item pair $(i,o) \in N \times M$,  the valuation $ v_i(o) \in  \{ \text{Discrete Uniform}(5,6), \text{Bernoulli}(0.5)\} $ is randomly generated.\footnote{Discrete Uniform(5, 6) meaning we uniformly sampled from the set $\{5,6\}$ and Bernoulli(0.5) means values are sampled from the set \{0,1\} each with probability $0.5$.}  We assume the additive valuations.

 For each realization of the random instance, we solved the following integer linear programming (ILP) problem using the Gurobi Optimizer solver (version 11.0.3) to compute the minimum subsidy and also we computed the total subsidy obtained by our Algorithms~\ref{alg:general-additive}-\ref{alg:n-subsidy-additive-identical-items}. We repeated the experiment $50$ times and reported the average total  subsidy in Table~\ref{tab:Alg1_averagedover50instances_disunif(5,6)} - \ref{tab:Alg4_averagedover50instances_disunif(5,6)}. 
\begin{gather*}
\min \sum_{i \in N} p_i \\
\begin{aligned}
\textup{s.t.}\quad\sum_{o \in M} \frac{v_i(o)x_{i,o}+p_i}{w_i}  &\geq  \sum_{o \in M} \frac{v_i(o)x_{j,o}+p_j}{w_j} ~\forall i, j \in N\\
\quad \sum_{i \in N}x_{i,o}&= 1 ~\forall o \in M \\
                  x_{i,o}  &  \in\{0,1\}~ \forall o \in M \\
        p_i  &\geq  0 ~ \forall i \in N.
\end{aligned}
\end{gather*}

\begin{table}[H]
\caption{The table shows the minimum subsidies obtained by solving the ILP problem, subsidies obtained by Algorithm~\ref{alg:general-additive}, and subsidies theoretically guaranteed by Algorithm~\ref{alg:general-additive} for $n\in \{5,8,10\}$ respectively.  The valuation function is $ v_i(o) \sim \text{Discrete Uniform}(5,6)$.}
\label{tab:Alg1_averagedover50instances_disunif(5,6)}
\begin{tabular}{|c|c|c|c|}
\hline
\textbf{Number of} & \multicolumn{3}{|c|}{\textbf{Total subsidy}} \\ \cline{2-4}
\textbf{items} & \textbf{Algorithm~\ref{alg:general-additive}}  & \textbf{Minimum} & \textbf{Theoretical bound} \\ \hline
5  & 62.5  & 10.3  & 84     \\ \hline
10 & 35.02 & 7.615 & 84     \\ \hline
15 & 7.84  & 2.085 & 84     \\ \hline
20 & 55.06 & 3.42  & 84     \\ \hline
25 & 29.2  & 5.555 & 84     \\ \hline
\end{tabular}

\vspace{0.2cm}
\begin{tabular}{|c|c|c|c|}
\hline
\textbf{Number of} & \multicolumn{3}{|c|}{\textbf{Total subsidy}} \\ \cline{2-4}
\textbf{items} & \textbf{Algorithm~\ref{alg:general-additive}}  & \textbf{Minimum} & \textbf{Theoretical bound} \\ \hline
8 & 171.7800 & 15.78   & 210    \\ \hline
16 & 128.24   & 7.9529  & 210    \\ \hline
24 & 84.06    & 9.6214  & 210    \\ \hline
32 & 40.08    & 4.5557  & 210    \\ \hline
40 & 176.1    & 3.7569  & 210    \\ \hline
\end{tabular}

\begin{tabular}{|c|c|c|c|}
\hline
\textbf{Number of} & \multicolumn{3}{|c|}{\textbf{Total subsidy}} \\ \cline{2-4}
\textbf{items} & \textbf{Algorithm~\ref{alg:general-additive}}  &\textbf{ Minimum} & \textbf{Theoretical bound} \\ \hline
 10 & 275      & 14.7583 & 324    \\ \hline
 20 & 220.24   & 14.3654 & 324    \\ \hline
 30 & 165.06   & 11.476  & 324    \\ \hline
 40 & 109.9    & 10.2745 & 324    \\ \hline
50 & 55.06    & 3.5088  & 324    \\ \hline
\end{tabular}
\end{table}


\begin{table}[H]
\caption{The table shows the minimum subsidies obtained by solving the ILP problem, subsidies obtained by Algorithm~\ref{alg:one-subsidy-identical-additive}, and subsidies theoretically guaranteed by Algorithm~\ref{alg:one-subsidy-identical-additive} for $n\in \{5,8,10\}$ respectively. The valuation function is $ v(o) \sim \text{Discrete Uniform}(1,2)$.}
\label{tab:Alg2_averagedover50instances_disunif(1,2)}
\begin{tabular}{|c|c|c|c|}
\hline
\textbf{Number of} & \multicolumn{3}{|c|}{\textbf{Total subsidy}} \\ \cline{2-4}
\textbf{items} & \textbf{Algorithm~\ref{alg:one-subsidy-identical-additive}}  & \textbf{Minimum} & \textbf{Theoretical bound }\\ \hline
5 & 3.515 & 2.17  & 8      \\ \hline
10 & 4.24  & 1.455 & 8      \\ \hline
15 & 3.85  & 2.005 & 8      \\ \hline
20 & 4.02  & 1.68  & 8      \\ \hline
25 & 4.205 & 2.14  & 8      \\ \hline
\end{tabular}
\vspace{0.2cm}
\begin{tabular}{|c|c|c|c|}
\hline
\textbf{Number of} & \multicolumn{3}{|c|}{\textbf{Total subsidy}} \\ \cline{2-4}
\textbf{items} & \textbf{Algorithm~\ref{alg:one-subsidy-identical-additive} } & \textbf{Minimum} & \textbf{Theoretical bound} \\ \hline
8  & 6.5531 & 3.9654 & 14     \\ \hline
16 & 6.9571 & 3.6863 & 14     \\ \hline
24 & 7.7911 & 2.1414 & 14     \\ \hline
32 & 6.0966 & 3.9274 & 14     \\ \hline
40 & 6.6254 & 3.8506 & 14    \\ \hline
\end{tabular}

\begin{tabular}{|c|c|c|c|}
\hline
\textbf{Number of} & \multicolumn{3}{|c|}{\textbf{Total subsidy}} \\ \cline{2-4}
\textbf{items} & \textbf{Algorithm~\ref{alg:one-subsidy-identical-additive}}  & \textbf{Minimum }& \textbf{Theoretical bound} \\ \hline
 10 & 8.5921 & 4.4552 & 18     \\ \hline
 20 & 9.5916 & 4.5768 & 18     \\ \hline
30 & 8.9475 & 5.843  & 18     \\ \hline
 40 & 9.1292 & 3.9327 & 18     \\ \hline
 50 & 8.8797 & 4.9171 & 18     \\ \hline
\end{tabular}
\end{table}


\begin{table}[H]
\caption{The table shows the minimum subsidies obtained by solving the ILP problem, subsidies obtained by Algorithm~\ref{alg:binary-additive}, and subsidies theoretically guaranteed by Algorithm~\ref{alg:binary-additive} for $n\in \{5,8,10\}$ respectively.  The valuation function is $ v_i(o) \sim \text{Bernoulli}(0.5)$.}
\label{tab:Alg3_averagedover50instances_Ber(0.5)}
\begin{tabular}{|c|c|c|c|}
\hline
\textbf{Number of} & \multicolumn{3}{|c|}{\textbf{Total subsidy}} \\ \cline{2-4}
\textbf{items} & \textbf{Algorithm~\ref{alg:binary-additive}}  & \textbf{Minimum} & \textbf{Theoretical bound} \\ \hline
5  & 1.69033  & 1.48699 & 14     \\ \hline
10 & 0.98299  & 0.23533 & 14     \\ \hline
15 & 0.370666 & 0       & 14     \\ \hline
20 & 0.29333  & 0       & 14     \\ \hline
25 & 0.422    & 0       & 14     \\ \hline
\end{tabular}
\vspace{0.2cm}
\begin{tabular}{|c|c|c|c|}
\hline
\textbf{Number of} & \multicolumn{3}{|c|}{\textbf{Total subsidy}} \\ \cline{2-4}
\textbf{items} & \textbf{Algorithm~\ref{alg:binary-additive}}  & \textbf{Minimum} & \textbf{Theoretical bound }\\ \hline
8  & 3.1364   & 2.4306 & 35     \\ \hline
16 & 1.8120   & 0.5452 & 35     \\ \hline
24 & 1.0444   & 0.0688 & 35     \\ \hline
32 & 1.1500   & 0 & 35     \\ \hline
40 & 0.2393   & 0 & 35     \\ \hline
\end{tabular}

\begin{tabular}{|c|c|c|c|}
\hline
\textbf{Number of} & \multicolumn{3}{|c|}{\textbf{Total subsidy} }\\ \cline{2-4}
\textbf{items} & \textbf{Algorithm~\ref{alg:binary-additive} } &\textbf{ Minimum} & \textbf{Theoretical bound} \\ \hline
 10 & 3.5305   & 3.0417 & 44     \\ \hline
 20 & 3.9967   & 0.7505 & 44     \\ \hline
 30 &  1.9807   &  0.1652      & 44     \\ \hline
 40 & 0.9708   & 0 & 44     \\ \hline
 50 & 2.2950   & 0 & 44     \\ \hline
\end{tabular}
\end{table}


\begin{table}[H]
\caption{The table shows the minimum subsidies obtained by solving the ILP problem, subsidies obtained by Algorithm~\ref{alg:n-subsidy-additive-identical-items}, and subsidies theoretically guaranteed by Algorithm~\ref{alg:n-subsidy-additive-identical-items} for $n\in \{5,8,10\}$ respectively. The valuation function is $ v_i(o) \sim \text{Discrete Uniform}(5,6)$. Here, an agent has valuation identical for all items.}
\label{tab:Alg4_averagedover50instances_disunif(5,6)}
\begin{tabular}{|c|c|c|c|}
\hline
\textbf{Number of }& \multicolumn{3}{|c|}{\textbf{Total subsidy}} \\ \cline{2-4}
\textbf{items }& \textbf{Algorithm~\ref{alg:n-subsidy-additive-identical-items}}  & \textbf{Minimum} & \textbf{Theoretical bound} \\ \hline
5  & 70.8417 & 26.1153 & 169.5  \\ \hline
10 & 98.3267 & 23.2863 & 169.5  \\ \hline
15 & 85.0533 & 0       & 169.5  \\ \hline
20 & 98.8933 & 25.0247 & 169.5  \\ \hline
25 & 102.16  & 22.494  & 169.5  \\ \hline
\end{tabular}
\vspace{0.2cm}
\begin{tabular}{|c|c|c|c|}
\hline
\textbf{Number of} & \multicolumn{3}{|c|}{\textbf{Total subsidy}} \\ \cline{2-4}
\textbf{items} & \textbf{Algorithm~\ref{alg:n-subsidy-additive-identical-items} } & \textbf{Minimum} & \textbf{Theoretical bound} \\ \hline
8  & 228.1196 & 51.0788 & 497.0574 \\ \hline
16 & 265.4938 & 59.0252 & 497.0574 \\ \hline
24 & 274.1384 & 50.6484 & 497.0574 \\ \hline
32 & 324.5231 & 20      & 497.0574 \\ \hline
40 & 344.4849 & 45.7897 & 497.0574 \\ \hline
\end{tabular}

\begin{tabular}{|c|c|c|c|}
\hline
\textbf{Number of} & \multicolumn{3}{|c|}{\textbf{Total subsidy}} \\ \cline{2-4}
\textbf{items} & \textbf{Algorithm~\ref{alg:n-subsidy-additive-identical-items}}  & \textbf{Minimum} & \textbf{Theoretical bound }\\ \hline
 10 & 374.8001 & 75.628   & 825.5598 \\ \hline
 20 & 413.9721 & 104.0536 & 825.5598 \\ \hline
 30 & 489.8345 & 93.5287  & 825.5598 \\ \hline
 40 & 496.2941 & 68.3228  & 825.5598 \\ \hline
 50 & 529.3542 & 25       & 825.5598 \\ \hline
\end{tabular}
\end{table}

From Tables~\ref{tab:Alg1_averagedover50instances_disunif(5,6)}-\ref{tab:Alg4_averagedover50instances_disunif(5,6)}, 
 we observe that our algorithm consistently provides a lower average subsidy than theoretically obtained bound, while it is larger than the ILP computed minimum subsidy.

\section{Future Work}
\label{sec:conclusions}
Although several important techniques from the unweighted setting do not work in the weighted setting, we have managed to develop new techniques and used them to prove that WEF allocations with subsidies can be computed for agents with general monotone valuations.
We even proved a worst-case upper bound on the amount of subsidy required to attain WEF.
Our bounds are tight for general monotone valuations, superadditive and supermodular valuations, and identical additive valuations;
however, for additive and binary valuations our bounds are not tight.
Tightening these bounds is one of the main problems left open by the present paper.

Preliminary simulation experiments (\Cref{sec:exp}) show that our algorithms require less subsidy than the worst-case bound, but more than the optimal amount. We plan to perform more comprehensive experiments in the future.

Moreover, it remains an open problem to find a WEF-able allocation for additive valuations when the ratios between the agents' weights are non-integer, as well as to find a WEF-able allocation for identical items with a minimum subsidy.

\newpage
\bibliographystyle{ACM-Reference-Format}
\bibliography{sample-bibliography}

\appendix
\section*{APPENDIX}
\section{Tightness of the Subsidy Bounds}
\subsection{Subsidy Bound of \Cref{alg:general-additive}} \label{alg:general-additive tightness}
     As \Cref{theorem: sub general additive} implies, \Cref{alg:general-additive} computes a WEF-able allocation with a total subsidy of at most $(W-\wmin)V$. However, this
     bound is not tight. To understand why, consider the case of $2$ items, each valued at $V$ by agent $i\in \{1, \ldots, n-1\}$, who has an entitlement of $w_i \geq 2$, and $V - \epsilon$ by all other agents. Our algorithm will allocate all the items to agent $i$, resulting in a subsidy of $\frac{w_j}{w_i} 2 (V-\epsilon)$ by each other agent $j\neq i \in N$, leading to a total subsidy of $(W-w_i)\frac{2(V-\epsilon)}{w_i}$, for arbitrarily small $\epsilon>0$.
     
     In general, a WEF-able allocation can achieve a lower subsidy by allocating one item to another agent with higher index $j > i$, i.e, $w_j \geq w_i$. For instance, if one item is allocated to such agent $j$, agent $i$ envies agent $j$ by $\frac{V}{w_j} - \frac{V}{w_i} \leq 0$, and agent $j$ envies agent $i$ by $\frac{V - \epsilon}{w_i} - \frac{V - \epsilon}{w_j} < \frac{2(V - \epsilon)}{w_i}$. If $\frac{V - \epsilon}{w_i} - \frac{V - \epsilon}{w_j} \leq 0$, then no subsidy is required. Otherwise, the subsidy required by agent $j$ is $\left( \frac{V - \epsilon}{w_i} - \frac{V - \epsilon}{w_j} \right) w_j < \frac{w_j}{w_i} \cdot 2(V - \epsilon)$. The subsidy required by each other agent $k \neq i,j$ is significantly lower than $w_k \cdot \frac{w_j}{w_j} \cdot \frac{2(V - \epsilon)}{w_i} = \frac{w_k}{w_i} \cdot 2(V - \epsilon)$. Therefore, the required total subsidy is significantly lower than $(W - w_i) \frac{2(V - \epsilon)}{w_i}$.

     In both cases, the resulting total subsidy bound is better than the bound obtained by allocating all items to agent $i$.

     \subsection{Subsidy Bound of Algorithm~\ref{alg:one-subsidy-identical-additive}} \label{alg:identical-additive tightness}

\begin{example}
    Consider 2 agents with weights $w_1 = 1$ and $w_2 = 2$, and 2 items $o_1$ and $o_2$, where $v(o_1)=\frac{V}{2}$ and $v(o_2)=V$.
    The allocation $X$ output by Algorithm~\ref{alg:one-subsidy-identical-additive} is $X = (\emptyset, \{o_1,o_2\})$ with total subsidy $\frac{3}{4}V$.
    On the other hand, total subsidy needed for allocation $(\{o_1\},\{o_2\})$ is zero.
    Hence, Algorithm~\ref{alg:one-subsidy-identical-additive} does not necessarily output an allocation with minimum total subsidy and the upper bound on the total subsidy given in Theorem~\ref{thm:poly-on-subsidy-identical-additive} is not always optimal.
 \end{example}

 \subsection{Subsidy Bound of \Cref{alg:binary-additive}} \label{alg:binary-additive tightness}
As \Cref{app:theorem_21} implies, \Cref{alg:binary-additive} computes a WEF-able allocation with a total subsidy of at most $\frac{W}{\wmin} - 1$. However, with more careful analysis, we can prove a tighter bound.

There are two cases to consider:
\begin{enumerate}
    \item \textbf{Agent 1 with the minimum entitlement receives a positive subsidy.} Together, \Cref{app:Alg4_sub_in_game} and \Cref{subsidy of agent not in the game} imply that $p_i \leq \frac{w_i}{\wmin}$ for each agent $i \in N$. Since agent 1 does receive a positive subsidy, and by \Cref{thm:minsubsidy}, there exists at least one agent who requires no subsidy, the total required subsidy is bounded by $\frac{W - w_2}{\wmin}$.
    \item \textbf{Agent 1 with the minimum entitlement receives no subsidy.} We can modify \Cref{subsidy of agent not in the game} in the following way: for each agent $i \notin R(t)$, where $t \in [T]$, $\ell_i(X^t) \leq \frac{1}{w_2}$. By the proof of \Cref{subsidy of agent not in the game}, $\ell_i(X^t) \leq \ell_{i_{j+1}}(X^t)$. If $i_{j+1} = i_1$, then $\ell_i(X^t) \leq \ell_{i_{j+1}}(X^t) \leq 0$ (because agent 1 requires no subsidy). Otherwise, $\ell_i(X^t) \leq \ell_{i_{j+1}}(X^t) \leq \frac{1}{w_{i_{j+1}}} \leq \frac{1}{w_2}$.
    Overall, the subsidy required by each agent is bounded by $\frac{w_i}{w_2}$, and by \Cref{thm:minsubsidy}, there exists at least one agent who requires no subsidy. Therefore, the total required subsidy is bounded by $\frac{W - \wmin}{w_2}$.
\end{enumerate}
To sum up, the total required subsidy is at most $\max\Big\{\frac{W - \wmin}{w_2}, \frac{W - w_2}{\wmin}\Big\}$.

\subsection{Subsidy Bound of Algorithm~\ref{alg:n-subsidy-additive-identical-items}} \label{alg:identical-items tightness}
The following example shows that \Cref{alg:n-subsidy-additive-identical-items} does not necessarily output an allocation with minimum total subsidy:
\begin{example}
    Consider 3 agents with weights $w_1 = w_2 = w_3 = 1$ and 4 identical items, where $v_1 > v_2 > v_3$.
    The allocation $X$ output by Algorithm~\ref{alg:n-subsidy-additive-identical-items} is $X = (2, 1, 1)$ with total subsidy $2v_2$.
    On the other hand, total subsidy needed for allocation $(2,2,0)$ is $2v_1$, which is smaller than $2v_2$.
 \end{example}
 This raises an important question that remains open: with identical items and additive valuations, is it possible to find a WEF allocation with minimum subsidy in polynomial time?
\end{document}